\documentclass[11pt]{article}

\usepackage[margin=1in]{geometry}
\usepackage{mathpazo}
\usepackage{stmaryrd}
\usepackage[english]{babel}
\usepackage[autostyle, english = american]{csquotes}
\usepackage{anyfontsize}
\usepackage{authblk}

\usepackage{tikz}

\usetikzlibrary{
  arrows,
  automata,
  shapes,
  shapes.geometric,
  positioning,
  calc,
  chains,
  decorations.pathmorphing,
  intersections, fit,cd}

\usepackage{amssymb, amsthm, amsmath, dsfont, mathtools}
\mathtoolsset{showonlyrefs}
\usepackage{latexsym}
\usepackage{eepic}
\usepackage{sectsty}
\usepackage{framed}
\usepackage[shortlabels]{enumitem}
\usepackage{float}
\restylefloat{table}
\usepackage[pagecolor=white]{pagecolor}
\usepackage{xcolor}
\makeatletter
\newcommand{\globalcolor}[1]{%
  \color{#1}\global\let\default@color\current@color
}
\makeatother

\AtBeginDocument{\globalcolor{black}}

\usepackage[backref=page]{hyperref}
\hypersetup{
    colorlinks = true,
    linkcolor = magenta,
    citecolor = blue,
    urlcolor = blue
}

\hypersetup{pdfpagemode=UseNone}

\newtheorem{theorem}{Theorem}
\newtheorem{lemma}[theorem]{Lemma}
\newtheorem{prop}[theorem]{Proposition}
\newtheorem{cor}[theorem]{Corollary}
\newtheorem{obs}[theorem]{Observation}

\newtheorem{question}[theorem]{Question}
\newtheorem{problem}[theorem]{Problem}

\theoremstyle{definition}
\newtheorem{definition}[theorem]{Definition}

\newtheorem{remark}[theorem]{Remark}
\newtheorem{fact}[theorem]{Fact}
\newtheorem{example}[theorem]{Example}

\newcommand{\tinyspace}{\mspace{1mu}}

\newcommand{\rank}{\operatorname{rank}}

\newcommand{\abs}[1]{\lvert #1 \rvert}

\newcommand{\ip}[2]{\langle #1 , #2\rangle}

\newcommand{\ceil}[1]{\lceil #1 \rceil}

\newcommand{\floor}[1]{\lfloor #1 \rfloor}

\newcommand{\norm}[1]{\lVert\tinyspace #1 \tinyspace\rVert}

\newcommand{\I}{\mathds{1}}

\newcommand{\setft}[1]{\mathrm{#1}}

\newcommand{\complex}{\mathbb{C}}
\newcommand{\field}{\mathbb{F}}
\newcommand{\real}{\mathbb{R}}
\renewcommand{\natural}{\mathbb{N}}

\newcommand\X{\mathcal{X}}
\newcommand\Y{\mathcal{Y}}
\newcommand\Z{\mathcal{Z}}
\newcommand\W{\mathcal{W}}
\newcommand\A{\mathcal{A}}
\newcommand\B{\mathcal{B}}
\newcommand\V{\mathcal{V}}
\newcommand\U{\mathcal{U}}
\newcommand\C{\mathcal{C}}
\newcommand\D{\mathcal{D}}

\newcommand\T{\mathcal{T}}
\newcommand\K{\mathcal{K}}

\renewcommand\L{\mathcal{L}}

\DeclareMathOperator{\spn}{span}

\usepackage{xifthen,mathrsfs}

\newcommand{\eql}[2]{\begin{align}\label{#1}#2\end{align}}
\newcommand{\eq}[2]{
\ifthenelse{\equal{#1}{}}{\begin{align}#2\end{align}}{\eql{#1}{#2}}}

\newcommand{\ha}[2][]{
\ifthenelse{\equal{#1}{}}{#2}{#1, #2}}

\newcommand{\hook}{\mathbin{\lrcorner}}
\newcommand{\codim}{\setft{codim}}
\newcommand{\trouble}{T}

  \newcommand{\anote}[1]{}
 \newcommand{\bnote}[1]{}
 \newcommand{\nnote}[1]{}

\def\ba#1\ea{\begin{align}#1\end{align}}

\begin{document}

\emergencystretch 3em
\title{\bf
Computing linear sections of varieties: quantum entanglement, tensor decompositions and beyond
}

 \author[$\dagger$]{Nathaniel Johnston\thanks{emails: njohnston@mta.ca, benjamin.lovitz@gmail.com, aravindv@northwestern.edu}}
 \author[$* \ddagger$]{
   Benjamin Lovitz}
  \author[$* \S$]{Aravindan Vijayaraghavan}
  \affil[$\dagger$]{Department of Mathematics and Computer Science, Mount Allison University, Sackville, New Brunswick, Canada}
 \affil[$\ddagger$]{Department of Mathematics, Northeastern University, Boston, Massachusetts, USA}
 \affil[$\S$]{Department of Computer Science, Northwestern University, Evanston, Illinois, USA}

\date{\today}

\maketitle
\begin{abstract}

We study the problem of finding elements in the intersection of an arbitrary conic variety in $\field^n$ with a given linear subspace (where $\field$ can be the real or complex field). This problem captures a rich family of algorithmic problems under different choices of the variety. The special case of the variety consisting of rank-$1$ matrices already has strong connections to central problems in different areas like quantum information theory and tensor decompositions. 
This problem is known to be NP-hard in the worst case, even for the variety of rank-$1$ matrices. 

In this work, we propose and analyze an algorithm for solving this problem. \anote{5/5: should we skip saying that algorithm based on simultaneous decomposition+Nullstellensatz?} 
Surprisingly, despite the above hardness results we show that our algorithm solves this problem efficiently for ``typical" subspaces. Here, the subspace $\U \subseteq \field^n$ is chosen \emph{generically} of a certain dimension, potentially with some {generic} elements of the variety contained in it. 
Our main result is a guarantee that our algorithm recovers all the elements of $\U$ that lie in the variety, under some mild non-degeneracy assumptions on the variety. 
As corollaries, we obtain the following new results:
\begin{itemize}
    \item Polynomial time algorithms for several \emph{entangled subspaces} problems in quantum entanglement, including determining $r$-entanglement, complete entanglement, and genuine entanglement of a subspace. While all of these problems are NP-hard in the worst case, our algorithm solves them in polynomial time for generic subspaces of dimension up to a constant multiple of the maximum possible.
        \item Uniqueness results and polynomial time algorithmic guarantees for {generic} instances of a broad class of low-rank decomposition problems that go beyond tensor decompositions. Here, we recover a decomposition of the form $\sum_{i=1}^R v_i \otimes w_i$, where the $v_i$ are elements of the given variety $\X$. This implies new uniqueness results and genericity guarantees even in the special case of tensor decompositions.
\end{itemize}
\end{abstract}
\newpage
\tableofcontents
\newpage
\section{Introduction}\label{intro}
%
%
%

Consider an $n$-dimensional vector space $\V$ over a field $\field$ that is either $\real$ or $\complex$. An (algebraic) variety $\X \subset \V$ is {\em cut out} by a collection of polynomials $f_1,\dots, f_p$, i.e. it is given by the common zeroes 
$$
\X=\{x \in \V : f_1(x)=0, f_2(x)=0, \dots, f_p(x)=0\}.
$$

We study the problem of finding points in the intersection of the given algebraic variety $\X$ with a linear subspace $\U$. The subspace $\U$ is specified by some basis $\{u_1, \dots, u_{R}\} \subseteq \V$, while the variety $\X$ is specified by a set of polynomials that cut it out. We will focus on the general class of \textit{conic} varieties, which are those  that are closed under scalar multiplication.\footnote{A conic variety is the affine cone over a projective variety. These results can equivalently be formulated in terms of projective varieties.} Conic varieties are cut out by homogeneous polynomials, which can be chosen to all have the same degree $d$.

\begin{problem}\label{problem:main}
Given as input a subspace $\U \subseteq \V$ specified by a basis $\{u_1, \dots, u_R\}$, and an arbitrary conic variety $\X\subseteq \V$ cut out by homogeneous degree-$d$ polynomials $f_1, \dots, f_p$, can we either certify that ${\U \cap \X = \{0\}}$ or else find a non-zero point $v \in \U \cap \X$?
\end{problem}

The above question encompasses a natural class of algorithmic problems that vary with the different choices of the variety. Even the special case of \textit{determinantal varieties} i.e., varieties of matrices of bounded rank, has rich connections to central problems in diverse areas such as quantum information theory and tensor decompositions. The set of $n_1 \times n_2$ matrices of rank at most $1$ forms a determinantal variety cut out by homogeneous polynomials of degree $2$ (corresponding to the determinants of all $2 \times 2$ submatrices being $0$). More generally, the set of matrices of rank at most $r$ forms a determinantal  variety cut out by polynomials of degree $r+1$.
Problem~\ref{problem:main} has the following applications in the context of quantum entanglement and low-rank decompositions (even for the special case of determinantal varieties): 
\begin{itemize}
\item In a bipartite quantum system, an {\em entangled subspace} is a linear subspace $\U$ of matrices that contains no product state, i.e no rank-$1$ matrix. 
Entangled subspaces have applications to certifying entanglement of mixed states~\cite{Hor97,BDMSST99}, constructing entanglement witnesses~\cite{ATL11,CS14}, certifying that a set of vectors is an unextendible product basis \cite{BDMSST99,DMSST03}, and designing quantum error correcting codes \cite{GW07,HG20}. An important algorithmic question in this context is determining whether a given subspace is entangled  ~\cite{Par04,Bha06}.
This algorithmic problem is a special case of  Problem~\ref{problem:main}, and is already NP-hard in the worst case~\cite{buss1999computational}.
Measuring and certifying other notions of entanglement e.g., $r$-entanglement, complete entanglement, genuine entanglement  are also captured by Problem~\ref{problem:main} for different choices of varieties.
    \item A \textit{rank $R$ decomposition} of a tensor $T$ is an expression of $T$ as a sum of $R$ rank-$1$ tensors. The \textit{tensor rank} of $T$ is the smallest integer for which a decomposition of that rank exists for $T$. The algorithmic goal in tensor decompositions is to find a rank $R$ decomposition of a given tensor $T$ if it exists. While this problem is NP-hard in the worst-case~\cite{HL}, there exist polynomial time algorithms that work for a broad range of the rank $R$ tensors on {\em generic} instances of the problem (i.e. the algorithm is successful on all but a zero measure set of instances). 
   The key subroutine in a famous algorithm due to De Lathauwer, Cardoso, and Castaing ~\cite{de2006link, DLCC} tries to find all the rank-$1$ matrices in a certain generic subspace, and is an instantiation of Problem~\ref{problem:main}.
\end{itemize}

In light of the computational intractability of Problem~\ref{problem:main}, our goal is to design algorithms with polynomial time guarantees for ``typical'' or {\em generic} instances. 
It is well known that a  generic linear subspace $\U$ of sufficiently small dimension $R$ (depending on the \textit{Krull dimension} of $\X$) does not contain any elements of the conic variety $\X$. 
%
%
For example, in the case of $n_1 \times n_2$ dimensional matrices, a generic linear subspace of dimension $R_0 \leq (n_1-1)(n_2-1)$  does not contain any rank-$1$ matrix almost surely~\cite{harris2013algebraic,CMW08}. 
Hence, if we consider a generic ${R \le R_0}$-dimensional subspace that contains $s \le R$ generic rank-1 matrices, we can hope to recover all of these $s$ planted elements.

In this work, we propose and analyze a polynomial-time algorithm (which we call \textit{Algorithm 1}, see Section~\ref{sec:algorithm}) to recover all the elements of $\U$ that lie in the variety $\X$. In more details, on input a collection of homogeneous degree-$d$ polynomials cutting out $\X\subseteq \field^n$ and any basis for the linear subspace $\U \subseteq \field^n$, Algorithm 1 runs in $n^{O(d)}$ time and either outputs ``Fail," or else outputs a finite collection of elements of the intersection $\U \cap \X$, along with a $n^{O(d)}$-time certificate that these are the \textit{only} elements of $\U \cap \X$ (up to scale). 
To simplify the analysis, we ignore issues of numerical precision (formally, we prove polynomial time guarantees in the real model of computation, given access to a constant number of calls to an oracle to diagonalize polynomial-sized matrices\footnote{A $k \times k$ diagonalizable matrix can be diagonalized to precision $\epsilon$ in time $O(k^{\omega} \log^2(k/\epsilon))$, where $\omega$ is the exponent of matrix multiplication~\cite{9317903}. Our algorithm requires a constant number of diagonalizations to run the simultaneous diagonalization algorithm as a subroutine, which itself has been shown to be numerically stable under some natural conditions~\cite{GVX14, BCMV}.}). Our Algorithm~1 is based on Hilbert's Nullstellensatz and a ``lifted" version of the simultaneous diagonalization algorithm to the space of degree-$d$ polynomials. In the special case when $\X$ is the set of rank-one matrices, our algorithm captures the tensor decomposition method proposed in~\cite{de2006link}. 
Our main result (Theorem~\ref{thm:intro:intersection}) guarantees that Algorithm 1 is always correct (i.e., any output that is not ``Fail'' is guaranteed to be correct), and does not output ``Fail" almost surely when $\dim(\U)$ is small enough.\footnote{Note that assumptions on genericity and dimension of the subspace are necessary: If $\U$ is a worst-case input or $\dim(\U)$ is too large, then $\U \cap \X$ could have an infinite number of non-parallel elements.}

\subsection{Our results}
We will require two technical assumptions on the variety $\X$, which will be satisfied by many varieties of interest: We say $\X$ is \textit{irreducible} if it cannot be written as a union of smaller varieties, and we say that an irreducible variety $\X$ is {\em non-degenerate of order $\tilde{d}$}
if $\X$ has no equations in degree $\tilde{d}$. For example, $\X$ is non-degenerate of order 1 if $\spn(\X)=\V$.
We say that an element is \textit{generically chosen} if it is chosen from a Zariski open dense subset of the underlying instance space (this also commonly referred to as a~\textit{general} element). Proving that a property holds for a generically chosen element is a standard algebraic-geometric approach to showing that it holds almost surely over the underlying instance space; see Section~\ref{sec:MP:AG} for more details. 

The following result applies when the field $\field$ is either $\real$ or $\complex$. 
The notation $\X^{\times s} \times \V^{\times (R-s)}$ denotes the set of $R$-tuples of elements of $\V$, the first $s$ of which are chosen from $\X$.

%

\begin{theorem}\label{thm:intro:intersection}
Let $\X \subseteq \V = \field^n$ be an irreducible variety cut out by $p= \delta \binom{n+d-1}{d}$ linearly independent homogeneous degree-$d$ polynomials $f_1,\dots, f_p \in \field[x_1,\dots, x_n]_d$, for constants $d \geq 2$ and $\delta \in (0,1)$. Suppose furthermore that $\X$ is non-degenerate of order $d-1$. 
Then a linear subspace $\U \subseteq \V$ of dimension
\begin{equation}\label{eq:R_cd}
R \le \frac{1}{d!} \cdot  \delta   (n+d-1), 
\end{equation}
spanned by a generically chosen element of $\X^{\times s} \times \V^{\times (R-s)}$ for some $s \in \{0,1,\dots, R\}$, contains \emph{only} $s$ elements in its intersection with $\X$ (up to scalar multiples), and on input any basis of $\U$ our Algorithm 1 correctly outputs these elements in $n^{O(d)}$ time. When $s=0$, Algorithm 1 certifies that $\U \cap \X = \{0\}$ in $n^{O(d)}$ time.
\end{theorem}

We remark that in the above theorem, the choice of the variety $\X$ is arbitrary (subject to the irreducibility and non-degeneracy conditions), while the subspace $\U$ is chosen generically. The theorem shows that when $\delta = \Omega(1)$ (this is the parameter setting for many varieties of interest), we get genericity guarantees for $R$ going up to a constant fraction of the maximum possible dimension $n$.  
As stated in the theorem, our algorithm runs in polynomial time (in the dimension $n$) as long as $d$ is fixed. Note that our algorithm assumes knowledge of the coefficients of the $p$ homogenous degree-$d$ polynomials $f_1, \dots, f_p$, which in itself requires $p \binom{n+d-1}{d}$ time.

As alluded to earlier, it is classically well known that a linear subspace $\U\subseteq \V$ of dimension $R\leq \setft{codim}(\X)$ spanned by a generically chosen point in $\X^{\times s} \times \V^{\times (R-s)}$ contains only $s$ elements in its intersection with $\X$ (up to scalar multiples) when $\X$ is irreducible and non-degenerate of order 1~\cite[Theorem 4.6.14]{flenner1999joins}, \cite[Definition 11.2]{harris2013algebraic}. 
However, for a particular subspace $\U$, it is NP-hard in general to \textit{find} these elements of the intersection and to \textit{certify} that they are the only ones \cite{buss1999computational}. Despite this hardness result, our Algorithm 1 runs in polynomial time, and either outputs ``Fail," or else finds elements of the intersection and certifies that they are the only ones.
Theorem~\ref{thm:intro:intersection} guarantees that our algorithm will almost surely output the latter, provided that $\U \subseteq \V$ has dimension $R$ upper bounded by~\eqref{eq:R_cd}.\footnote{  
We remark that the righthand side of~\eqref{eq:R_cd} can be verified to be always less than or equal to $\setft{codim}(\X)$.} 
We call this a \textit{genericity guarantee} for Algorithm~1.

It is natural to ask if the irreducibility and non-degeneracy conditions can be removed. The \emph{irreducibility} condition can indeed be removed, by assuming the non-degeneracy condition holds for every irreducible component of $\X$.
The \emph{non-degeneracy} assumption can also removed if $s=0$ (i.e. the last sentence of the theorem holds without any non-degeneracy assumption nor irreducibility assumption on $\X$). See Corollary~\ref{cor:generic_simplified}. Some form of non-degeneracy assumption on $\X$ is necessary for general $s$: For example, if $\X$ is a linear subspace, then $\X$ can be cut out by degree-$2$ polynomials, but the intersection $\U \cap \X$ contains the entire span of $\{v_1,\dots, v_s\}$, so for $s \geq 2$ we cannot hope to recover $v_1,\dots, v_s$. (See also the discussion after Theorem~\ref{thm:decomp}).
%
%
Moreover, many commonly studied varieties satisfy this non-degeneracy assumption, as we will see below.

Consider the specific case of the variety of rank-$1$ matrices $\X_1=\{ M \in \field^{n_1 \times n_2}: \text{rank}(M)\leq 1 \}$. This is an irreducible variety that is cut out by $p=\binom{n_1}{2} \binom{n_2}{2}$ homogenous polynomials of degree $d=2$. 
Furthermore $\X_1$ is non-degenerate of order 1, i.e. $\text{span}(\X_1) = \field^{n_1\times n_2}$ . Hence we get the following immediate corollary, which already implies new results for quantum entanglement and tensor decompositions:


\begin{cor}\label{corr:intro:rank1}
A linear subspace $\U \subseteq \V=\field^{n_1 \times n_2}$ of dimension
\begin{equation}
R \le \frac{\binom{n_1}{2}  \binom{n_2}{2}}{2\binom{n_1 n_2+1}{2}} \cdot (n_1 n_2+1)= \frac{1}{4} (n_1-1) (n_2-1),
\end{equation}
spanned by a generically chosen element of $\X_1^s \times \V^{\times (R-s)}$ for some $s \in \{0,1,\dots, R\}$, contains only $s$ elements in its intersection with $\X_1$ (up to scalar multiples), and our Algorithm 1 correctly outputs these elements in $(n_1n_2)^{O(1)}$ time. When $s=0$, Algorithm 1 certifies that $\U \cap \X_1=\{0\}$ in $(n_1n_2)^{O(1)}$ time.
\end{cor}

More generally, the set of matrices of rank at most $r$, $\X_r=\{ M \in \field^{n_1 \times n_2}: \text{rank}(M)\leq r \}$, forms an irreducible variety cut out by $p=\binom{n_1}{r+1} \binom{n_2}{r+1}$ homogenous polynomials of degree $d=r+1$, and is non-degenerate of order $r$. We thus obtain the following consequence of Theorem~\ref{thm:intro:intersection}:

\begin{cor}\label{corr:intro:rankr}
Let $r$ be a fixed positive integer, and let $n_1,n_2 > r$ be integers. Then for a linear subspace $\U \subseteq \V=\field^{n_1 \times n_2}$ of dimension
\begin{equation}
R \le \frac{\binom{n_1}{r+1}  \binom{n_2}{r+1}}{(r+1)! \binom{n_1 n_2+r}{r+1}} \cdot (n_1 n_2+r),~~ \Big(\text{note that }\frac{\binom{n_1}{r+1}  \binom{n_2}{r+1} (n_1 n_2+r)}{(r+1)!\binom{n_1 n_2+r}{r+1}}=\Omega_r(n_1 n_2) \Big),
\end{equation}
spanned by a generically chosen element of $\X_r^s \times \V^{\times (R-s)}$ for some $s \in \{0,1,\dots, R\}$, contains only $s$ elements in its intersection with $\X_r$ (up to scalar multiples), and our Algorithm 1 correctly outputs these elements in $(n_1n_2)^{O(1)}$ time. When $s=0$, Algorithm 1 certifies that $\U \cap \X_r=\{0\}$ in $(n_1n_2)^{O(1)}$ time.
\end{cor}

In the remainder of this introduction, we describe applications of our algorithm to quantum entanglement and low-rank decomposition problems over varieties.

%
%
%
%
%

\subsection{Entangled subspaces} \label{sec:intro:entangled}

In the context of quantum information theory, there are various choices of varieties $\X$ for which it is useful to determine whether or not a given linear subspace $\U$ intersects $\X$. For example, if $\field = \complex$ and $\V = \field^{n_1} \otimes \field^{n_2} \cong \field^{n_1 \times n_2}$ then the unit vectors in $\V$ are called \emph{pure quantum states}. The states in the variety $\X_1=\{M \in \V : \rank(M) \leq  1\}$ are called \emph{separable states}, while those in $\V\setminus \X_1$ are said to be \emph{entangled}. Entangled states are of central importance in this area, as they are required as a starting point for many quantum algorithms and protocols, like quantum teleportation \cite{BBCJPW93} and superdense coding \cite{BW92}. More generally, the states in the determinantal variety $\X_r=\{M \in \V : \rank(M)\leq r\}$ are said to have \emph{Schmidt rank} at most $r$, and this notion of rank is regarded as a rough measure of \emph{how} entangled the quantum state is \cite{NC00}.

A linear subspace $\U \subseteq \V$ in which every pure state is highly entangled (i.e., has Schmidt rank strictly larger than $r$) is called \emph{$r$-entangled} (or just \emph{entangled} if $r = 1$). Such subspaces have found an abundance of applications in quantum entanglement theory and quantum error correction \cite{Hor97,BDMSST99,ATL11,CS14,HM10}. Determining whether or not a subspace $\U$ is $r$-entangled is exactly Problem~\ref{problem:main} in the case of the variety $\X = \X_r$, and this problem is known to be NP-hard in the worst case, even for $r=1$~\cite{buss1999computational}. To our knowledge, the best known algorithm requires a certain $\epsilon$-promise and takes $\exp(\tilde{O}({\sqrt{n_1}/\epsilon}))$ time in the worst case when $r=1$ and $n_1=n_2$~\cite{barak2017quantum} (see Section~\ref{sec:entanglement} for more details). Existing algorithms for solving similar problems either lack complexity-theoretic guarantees or only work in limited situations, such as when the subspace's dimension is smaller than $\min\{n_1,n_2\}$ \cite{LPS06,GR08,bousse2018linear,DRA21}. Surprisingly, by Corollaries~\ref{corr:intro:rank1} and~\ref{corr:intro:rankr}, our algorithm solves this problem for generic instances of $\U$ (and does not require the $\epsilon$-promise), as long as $\dim(\U)$ is less than a constant fraction of the total dimension $n_1 n_2$. For example, when $r=1$ we obtain the following, which is just the $s=0, s=1$ cases of Corollary~\ref{corr:intro:rank1}:

\begin{cor}\label{corr:intro:cert_rank1}
Suppose $\field = \complex$ and let $\U \subseteq \field^{n_1} \otimes \field^{n_2}$ be a generically chosen linear subspace of dimension
\begin{equation}
R \le \frac{\binom{n_1}{2}  \binom{n_2}{2}}{2\binom{n_1 n_2+1}{2}} \cdot (n_1 n_2+1)= \frac{1}{4} (n_1-1) (n_2-1)
\end{equation}
with possibly a generically chosen planted separable state. Then, in $(n_1n_2)^{O(1)}$ time, our algorithm either certifies that $\U$ is entangled or else produces the planted separable state in $\U$ and a certificate that this is the only separable state in $\U$.
\end{cor}
Note that, even if $\U$ contains a separable state, generically this will be the \textit{only} separable state contained in $\U$, and our algorithm will certify this efficiently. In fact, even if $\U$ contains several (but less than $\dim(\U)$) generic separable states, our algorithm finds them all and certifies that there are no others. This extends the range of applications of our algorithm beyond situations where it is useful to know that a subspace is entangled, such as in the construction of an \textit{unextendible product basis} \cite{BDMSST99}, to situations where it is useful to know that that a subspace does not contain enough separable states to span it, such as in the construction of an \textit{uncompleteable product basis} \cite{DMSST03}.

More generally, we use Theorem~\ref{thm:intro:intersection} to obtain similar guarantees for our algorithm to determine whether a subspace exhibits other notions of entanglement, which corresponds to answering Problem~\ref{problem:main} for other varieties $\X$:
\begin{itemize}
\item When $\V=\field^{n_1}\otimes \dots \otimes \field^{n_m}$ and $\X=\X_{\setft{Sep}} \subseteq \field^{n_1}\otimes \dots \otimes \field^{n_m}$ is the set of \textit{separable tensors} (tensors of the form $v_1\otimes v_2 \otimes \dots \otimes v_m$), our algorithm determines in $O(n_1\cdots n_m)$ time whether $\U \cap \X_{\setft{Sep}}=\{0\}$ (i.e., whether $\U$ is \textit{completely entangled}) for generically chosen subspaces $\U$ of dimension up to a constant multiple of the total dimension $n_1 n_2 \cdots n_m$.
\item When $\X=\X_B$ is the set of \textit{biseparable tensors} (tensors which are rank $1$ with respect to one of the $2^{m-1}$ different ways to view a tensor $T \in \field^{n_1}\otimes \dots \otimes \field^{n_m}$ as a matrix by grouping factors), our algorithm determines in $O(2^m n_1\cdots n_m)$ time whether $\U \cap \X_{\setft{Sep}}=\{0\}$ (i.e., whether $\U$ is \textit{genuinely entangled}) for generically chosen subspaces $\U$ of dimension up to a constant multiple of the total dimension $n_1 n_2 \cdots n_m$.
\item As a final application, which does not necessarily directly apply to studying quantum entanglement, we use our algorithm to determine whether a subspace $\U$ intersects the variety $\X_S$ of tensors of \textit{slice rank 1} (see Section~\ref{sec:XR}). In this case, our algorithm determines in $O(m n_1\cdots n_m)$ time whether $\U \cap \X_{S}=\{0\}$ for generically chosen subspaces $\U$ of dimension up to a constant multiple of the total dimension $n_1 n_2 \cdots n_m$. The slice rank has recently arisen as a useful tool for studying basic questions in computer science such as the capset and sunflower problems~\cite{petrov2016combinatorial,kleinberg2016growth,blasiak2016cap,naslund2017upper,fox2017tight}.
\end{itemize}

Our algorithm generalizes a recent algorithm studied in~\cite{JLV22} for \textit{certifying} entanglement in a subspace, in two ways: First, our algorithm can not only certify that a subspace trivially intersects $\X$, but it can also produce elements of $\U \cap \X$ (if they exist) and prove that these are the only elements of $\U\cap \X$ in polynomial time. Second, our algorithm has provable genericity guarantees for arbitrary conic varieties $\X$ that satisfy the non-degeneracy assumption.

\subsection{Low-rank decompositions over varieties} \label{sec:intro:decomp}

Low-rank decompositions of matrices and tensors form a powerful algorithmic toolkit that are used in data analysis, machine learning and high-dimensional statistics. Consider a general decomposition problem, 
where we are given a tensor $T$ that has a rank-$R$ decomposition of the form
\begin{equation} \label{eq:intro:lowrank}
    T=\sum_{i=1}^{R} v_i \otimes w_i,
\end{equation}   
where $v_1, \dots, v_R$ lie in a variety $\X \subseteq \V$ and $w_1, \dots, w_R$ are arbitrary vectors in $\W$; here $\V$ and $\W$ are vector spaces over a field $\field$ (either $\real$ or $\complex$). 
The goal is to recover a rank-$R$ decomposition given $T$, and when possible recover the above decomposition. 
These $(\X, \W)$\textit{-decompositions}, also known as \textit{simultaneous $\X$-decompositions},
specialize to other well-studied decomposition problems such as block decompositions (see Sections~\ref{sec:tensor:related} and \ref{sec:tensors}).  
When $\X$ is the entire space $\V$, these are standard matrix decompositions. 
When $\X$ is the variety corresponding to rank-$1$ matrices (or more generally, rank-$1$ tensors), then this leads to the {\em tensor decomposition problem} where the decomposition has the form $\sum_{i=1}^R y_i \otimes z_i \otimes w_i$.\footnote{One can also get symmetric decompositions of the form $\sum_i^r u_i^{\otimes 3}$ (by restricting $y_i, z_i$ to be equal, and setting $w_i$ appropriately).} More generally, this gives a rich class of higher order decomposition problems depending on the choice of the variety. 

A remarkable property of low-rank tensor decompositions is that their minimum rank decompositions are often unique up to trivial scaling and relabeling of terms.  This is in sharp contrast to matrix decompositions, which are not unique for any rank $R \ge 2$.\footnote{For any matrix $M$ with a rank $R \ge 2$ decomposition $M=\sum_{i=1}^R v_i \otimes w_i$, there exist several other rank $R$ decompositions $\sum_i^R v'_i \otimes w'_i$ where $v'_i=O v_i$ and $w'_i = O w_i$  for any matrix $O$ with $O O^T=I_{R}$ (this is called the rotation problem).} 
The first uniqueness result for tensor decompositions was due to Harshman~\cite{harshman1970foundations} (who in turn credits it to Jennrich) --- 
if an $n \times n \times n$ tensor $T$ has a decomposition 
$T=\sum_{i=1}^R y_i \otimes z_i \otimes w_i$ for $R \le n$, then for generic choices of $\{y_i, z_i, w_i\}$ this is the {\em unique decomposition} of rank $R$ up to permuting the terms. 
Moreover, while computing the minimum rank decomposition is NP-hard in the worst-case~\cite{Has90, HL}, under the same genericity conditions as above there exists a polynomial time algorithm that recovers the decomposition~\cite{Har72, LRA93}. A rich body of subsequent work 
 gives stronger uniqueness results and algorithms  for tensor decompositions~\cite{kruskal1977three, CO12, de2006link,DLCC}.   
%
%
%
These efficient algorithms and uniqueness results for tensors are powerful algorithmic tools that have found numerous applications including efficient polynomial time algorithms for parameter estimation of latent variable models like mixtures of Gaussians, hidden Markov models, and even for learning shallow neural networks; see \cite{Moitrabook, Anandkumarbook, Vij20} for more on this literature. This prompts the following question:  

\begin{question}\label{qn:decomp}
When can we design efficient algorithms that achieve unique recovery for low-rank decomposition problems beyond tensor decompositions?
\end{question}

In answer to this question, we prove that one can establish uniqueness and efficiently recover decompositions of the form~\eqref{eq:intro:lowrank} for any irreducible conic variety satisfying the non-degeneracy assumption introduced above:

\begin{theorem}[Uniqueness and efficient algorithm for decompositions]\label{thm:decomp}
Let $\X \subseteq \V=\field^n$ be an irreducible conic variety cut out by $p=\delta \binom{n+d-1}{d}$ linearly independent homogeneous degree-$d$ polynomials for constants $d \ge 2$ and $\delta\in(0,1)$. Suppose furthermore that $\X$ is non-degenerate of order $d-1$. Then there is an $n^{O(d)}$-time algorithm that, on input a generically chosen tensor $T \in \V \otimes \W$
of $(\X,\W)$-rank 
\begin{equation}\label{eq:rank_bound}
R \leq \min\Big\{ \frac{1}{d!} \cdot \delta (n+d-1) , \dim(\W) \Big\},
\end{equation}
outputs an $(\X,\W)$-rank decomposition of $T$ and certifies that this is the unique $(\X,\W)$-rank decomposition of $T$.
\end{theorem}

Theorem~\ref{thm:decomp} follows from Theorem~\ref{thm:intro:intersection} by viewing $T$ as a map $T: \W^* \rightarrow \V$ and running Algorithm 1 on any basis of $\U=T(\W^*)$ (see Theorem~\ref{cor:generic_decomp} for details). 
A similar remark to that of the second paragraph following Theorem~\ref{thm:intro:intersection} is in order: The fact that a generically chosen tensor $T$ of $(\X,\W)$-rank upper bounded by~\eqref{eq:rank_bound} has a unique decomposition follows from known results~\cite[Theorem 4.6.14]{flenner1999joins}. The main contribution in this theorem is the genericity guarantee for our $n^{O(d)}$-time algorithm to \textit{recover} this decomposition and \textit{certify} that it is unique. Similar to Harshman's algorithm, our algorithm is accompanied by a concrete sufficient condition for a given $(\X,\W)$-decomposition to be unique (Proposition~\ref{prop:suffsimplified}).

As in Theorem~\ref{thm:intro:intersection}, our algorithm uses the description of the variety $\X$ as specified by the coefficients of the $p$ homogenous degree-$d$ polynomials that cut out $\X$. A {\em generically chosen} tensor ${T \in \V \otimes \W}$ of $(\X,\W)$-rank $R$ is formed by choosing $v_1,\dots, v_{R}$ generically from the variety $\X$, choosing $w_1, \dots, w_R$ generically from the vector space $\W$, and letting $T=\sum_{i=1}^{R} v_i \otimes w_i$. Note that when $R \le \dim(\W)$, almost surely the vectors $w_1, \dots, w_R$ are linearly independent. Due to the non-degeneracy of $\X$, Theorem~\ref{thm:decomp} also gives guarantees under the same rank condition for decompositions of the form:

$$ T = \sum_{i=1}^R v_i \otimes v'_i, \text{ where } v_1, \dots, v_R \text{ and } v'_1, \dots, v'_R \text{ are chosen {\em generically} from }  \X$$
(see Theorem~\ref{cor:generic_decomp} and the subsequent discussion).

%
%

\paragraph{Implications for tensor decompositions and beyond} While the above result holds for an arbitrary non-degenerate variety, even for the standard tensor decomposition problem where the tensor $T \in \field^{n_1} \otimes \field^{n_2} \otimes \field^{n_3}$ has the form
$$ T = \sum_{i=1}^R x_i \otimes y_i \otimes w_i,$$
we get improved guarantees by restricting our attention to the variety of rank-$1$ matrices in $\field^{n_1 \times n_2}$. 

\begin{cor}\label{cor:intro:tensor_decomp}
For any positive integers $n_1, n_2, n_3$, there is an $(n_1 n_2)^{O(1)}$-time algorithm that, on input a generically chosen tensor $T\in \field^{n_1}\otimes \field^{n_2} \otimes \field^{n_3}$ of tensor rank
\ba
R \leq \min\Bigg\{\frac{1}{4}(n_1-1) (n_2-1), n_3\Bigg\},
\ea
outputs a tensor rank decomposition of $T$ and certifies that this is the unique tensor rank decomposition of $T$.
\end{cor}
See Corollary~\ref{cor:tensor_decomp} in Section~\ref{sec:tensors} for the proof.
To interpret this result, consider the setting when $n_1 = n_2=n \le n_3$. Existing algorithms for order-$3$ tensors (e.g., \cite{Har72, LRA93, EVDL2022}) give genericity guarantees when $R \le n$. On the other hand, Corollary~\ref{cor:intro:tensor_decomp} gives guarantees for rank $\min\{(\tfrac{1}{4}-o(1)) n^2 , n_3\}$ which can be significantly larger -- for a tensor with $n_3=\Omega(n^2)$, we can even handle tensors of rank $R = \Omega(n^2)$, which is the best possible up to constants (an $n \times n \times n^2$ tensor has rank at most $O(n^2)$). 
A similar result to Corollary~\ref{cor:intro:tensor_decomp} was earlier claimed by \cite{de2006link}, but we show that a crucial lemma in their proof is incorrect. See Section~\ref{sec:tensor:related} for a more detailed description, and see Appendix~\ref{app:counterexample} for a counterexample to the lemma in question. This lemma was  also a crucial ingredient in the claimed proof appearing in~\cite{DLCC} that the FOOBI algorithm recovers symmetric decompositions of generically chosen $n \times n \times n \times n$ symmetric tensors of rank up to $O(n^2)$. Our analysis provides corrected proofs of these results using a different proof technique, and applies in a much more general setting. 

For higher order tensors, we obtain algorithms with polynomial time (unique) recovery guarantees for non-symmetric tensors. For generic order-$m$ tensors of rank $R =O(n^{\lfloor m/2 \rfloor})$, our algorithm finds the unique tensor decomposition of rank-$R$ in $n^{O(m)}$ time. (See Corollary~\ref{cor:general_tensor_decomp} for a slightly stronger statement, and Section~\ref{sec:tensors} for more results on tensors.) For even order $m$, we are not aware of prior works that prove such genericity guarantees for rank up to $n^{m/2}$;  such guarantees are known only in the special case of symmetric decompositions~\cite{MSS, BCPV}. See Section~\ref{sec:tensor:related} for related work on tensor decompositions. 

While these results give improvements even in the case of standard tensor decompositions, our algorithmic framework gives uniqueness results and efficient algorithms for a much broader class of low-rank decomposition problems. One such collection of applications are \textit{aided decompositions}, also known as~\textit{block decompositions}, which are generalizations of tensor decompositions that are useful in signal processing and machine learning~\cite{kolda2009tensor,comon2010handbook,cichocki2015tensor,sidiropoulos2017tensor,de2008decompositionspart1,de2008decompositionspart2,de2008decompositionspart3,domanov2020uniqueness}. Our general result (Theorem~\ref{thm:decomp}) also gives guarantees for such block decompositions; see Corollary~\ref{cor:block_decomp}. 


\subsection{Technical overview}\label{sec:overview}

Let $\V=\field^n$. Our main result is an algorithm for finding the points in the intersection of a conic variety $\X \subseteq \V$ with a linear subspace $\U$. 
Our algorithm is based on Hilbert's projective Nullstellensatz from algebraic geometry over $\complex$, as well as a ``lifted" version of the simultaneous diagonalization algorithm (also known as Jennrich's algorithm). 
\anote{5/5: rephrased mildly} Our algorithm generalizes the method of De Lathauwer, Cardoso and Castaing called ``FOOBI''~\cite{de2006link,DLCC}  for finding rank-one matrices in a subspace,\footnote{thereby also establishing genericity guarantees for the FOOBI algorithm and its variants~\cite{de2006link,DLCC}.} which is based on what they call a ``rank-$1$ detecting device'' $\Phi$. We will describe our algorithm using language that is more natural from an algebraic-geometric viewpoint.

We frequently invoke the canonical vector space isomorphism $\field[x_1,\dots, x_n]_d \cong S^d(\V^*)$ between the set of homogeneous polynomials of degree $d$ and the $d$-th symmetric power of $\V^*$, which sends a product of linear forms $f_1 \cdots f_d$ to the projection of $f_1\otimes \dots \otimes f_d$ onto $S^d(\V^*)$ (since $\setft{Char}(\field)=0$, one can take $f_1 \cdots f_d \mapsto \frac{1}{d!}\sum_{\sigma \in \mathfrak{S}_d} f_{i_{\sigma(1)}}\otimes \dots \otimes f_{i_{\sigma(d)}} \in S^d(\V^*)\subseteq (\V^*)^{\otimes d}$).

\subsubsection{The algorithm}

On input a set of homogeneous degree-$d$ polynomials $f_1,\dots, f_p \in S^d(\V^*)$ cutting out a variety $\X \subseteq \V$, and a basis $\{u_1,\dots, u_R\}$ for a linear subspace $\U\subseteq \V$, our algorithm first computes a basis $\{P_1,\dots, P_s\}$ for the linear subspace $S^d(\U) \cap I_d^{\perp} \subseteq S^d(\V)$, where $I_d=\spn\{f_1,\dots, f_p\}$, and $(\cdot)^\perp$ denotes the orthogonal complement in the dual space. Since $f_1,\dots, f_p$ cut out $\X$, the linear subspace $I_d^{\perp}\subseteq S^d(\V)$ satisfies the property that
\ba
\X=\{v \in \V : v^{\otimes d} \in I_d^{\perp}\}.
\ea
If $S^d(\U) \cap I_d^{\perp}=\{0\}$, then $\U\cap \X=\{0\}$, since for any element $v \in \U \cap \X$ it holds that $v^{\otimes d} \in S^d(\U) \cap I_d^{\perp}$.\footnote{This part of our algorithm, which checks whether $S^d(\U) \cap I_d^{\perp}=\{0\}$, constitutes a degree-$d$ \textit{Nullstellensatz certificate} for checking whether $\U \cap \X=\{0\}$ (see Remark~\ref{rmk:nullstellensatz}).} In this case, our algorithm outputs ``$\U$ trivially intersects $\X$."   

\anote{5/5: new para} Otherwise, our algorithm runs the simultaneous diagonalization algorithm on $\{P_1,\dots, P_s\}$ to determine whether there is a basis for $S^d(\U) \cap I_d^{\perp}$ of the form $\{v_1^{\otimes d},\dots, v_s^{\otimes d}\}$ for some linearly independent $v_1,\dots, v_s \in \V$. If this is the case, then clearly $v_1,\dots, v_s \in \U \cap \X$, and in fact the linear independence of $v_1,\dots, v_s$ guarantees that these are the \textit{only} elements of $\U \cap \X$ up to scale (see Observation~\ref{obs:find}). Accordingly, our algorithm outputs ``The only elements of $\U \cap \X$ are $\{v_1,\dots, v_s\}$ (up to scale)." If such a basis does not exist, our algorithm simply outputs ``Fail."
\anote{5/5: added} Note that in both of the above cases, the required condition on $S^d(\U) \cap I_d^{\perp}$ need not hold for certain choices of $\U$; as we describe in Section~\ref{sec:techniques:proof}, Theorem~\ref{thm:intro:intersection} establishes that this indeed holds for generic $\U$.

\anote{5/5: this comparison to FOOBI is quite long and somewhat early, and many readers may not even know what it is. Maybe we should move this to the end of the subsection, with a small paragraph called "Comparison to FOOBI", or even move it to the next section on Related work?}

\subsubsection{Proving the genericity guarantee} \label{sec:techniques:proof}

Our main technical contribution is Theorem~\ref{thm:intro:intersection}, which guarantees that our algorithm works extremely well on generic inputs: Assuming that $R$ is not too large (see~\eqref{eq:R_cd}), when $v_1,\dots, v_s \in \X$ and $v_{s+1},\dots, v_R \in \V$ are generically chosen, our algorithm recovers $v_1,\dots, v_s$ (up to scale) from any basis of $\U:=\spn\{v_1,\dots, v_R\}$ in $n^{O(d)}$ time, and certifies that these are the only elements of $\U \cap \X$. As can be readily verified from the description of our algorithm, this reduces to proving that for generically chosen $v_1,\dots, v_s \in \X$ and $v_{s+1},\dots, v_R \in \V$, it holds that $S^d(\U) \cap I_d^{\perp}\subseteq \spn\{v_1^{\otimes d},\dots, v_s^{\otimes d}\}$ (the reverse inclusion holds automatically).

For the proof, we in fact use nothing about the vector space $I_d^{\perp}\subseteq S^d(\V)$ other than its codimension $p$. To emphasize this, let $\K \subseteq S^d(\V)$ be an arbitrary linear subspace of codimension $p$. To give an idea of the proof, let us first consider the case when $s=0$. In this case we are setting out to prove that for a generic linear subspace $\U \subseteq \V$ of dimension not too large, it holds that $S^d(\U)\cap \K=\{0\}$. This seems to be an interesting multilinear algebraic question in its own right. Of course, a generically chosen subspace $\L \subseteq S^d(\V)$ of dimension $\dim(\L)=\codim(\K)$ will satisfy $\L \cap \K=\{0\}$, but our question is non-trivial because we constrain $\L$ to take the particular form $S^d(\U)$ for some $\U \subseteq \V$. The maximum dimension one could hope for would be the maximal $R$ for which $\binom{n+R-1}{R} \leq \codim(\K)$ (the lefthand side is $\dim(S^d(\U))$ when $\dim(\U)=R$). However, it is not always possible to achieve this.\footnote{Indeed, let $n=3$, $\V=\field^3$, and $\K=S^2(\field^2)\subseteq S^2(\field^3)$. Then for $R=2$ we have $\binom{n+R-1}{R}=6 = \codim(\K)$, but $S^2(\U)\cap \K \supseteq S^2( \U\cap \field^2 )\neq \{0\}$ for any $\U \subseteq \field^3$ of dimension $2$~\cite{415807}.}

More generally, for arbitrary $s \in \{0,1,\dots, R\}$, when $v_1,\dots, v_R$ are linearly independent (which holds generically since $\X$ is non-degenerate), the desired inclusion $S^d(\U) \cap \K\subseteq \spn\{v_1^{\otimes d},\dots, v_s^{\otimes d}\}$ is equivalent to
\ba
\spn\{P_{\V,d}^{\vee}(v_{a_1}\otimes \dots \otimes v_{a_d}) : a \in [R]^{\vee d} \setminus \Delta_s\} \cap \K=\{0\},
\ea
where $[R]^{\vee d}$ is the set of non-decreasing $d$-tuples of elements of $[R]:=\{1,\dots, R\}$, $\Delta_s=\{(1,1,\dots,1),\dots, (s,s,\dots,s)\}$, and $P_{\V,d}^{\vee}: \V^{\otimes d} \rightarrow S^{d}(\V)$ is the projection onto $S^{d}(\V)$.
A common approach in algebraic geometry for proving such a statement is to exhibit one choice of vectors $v_1, \dots, v_R$ for which it holds; this would imply the statement for a generic choice of $v_1, \dots, v_R$.
However, we do not know how to construct such vectors explicitly.

To prove the statement, we first define a certain total order $\succ$ on $[R]^{\vee d}$ for which the largest elements are $(1,\dots, 1)\succ \dots \succ (s,\dots, s)\succ \dots$. We then observe that it suffices to prove that, for all $(i_1,\dots, i_d) \in [R]^{\vee d} \setminus \Delta_s$, it holds that
\ba\label{eq:intro_simplified}
P_{\V,d}^{\vee}(v_{i_1}\otimes \dots \otimes v_{i_d})\notin \spn\{P_{\V,d}^{\vee}(v_{a_1}\otimes \dots \otimes v_{a_d}) : (i_1,\dots,i_d) \succ (a_1,\dots,a_d)\in [R]^{\vee d} \} + \K
\ea
for generically chosen $v_1,\dots, v_s \in \X, v_{s+1},\dots, v_R \in \V$. We prove this statement by induction on $d$. The crucial assumption is that $\X$ is non-degenerate of order $d-1$, or equivalently
\ba\label{eq:nondegenerate}
\spn\{v^{\otimes (d-\ell)} : v \in \X\}=S^{(d-\ell)}(\V)
\ea
for all $1 \leq \ell \leq d-1$. To prove the statement, we first note that the case $i_1=\dots=i_d>s$ follows easily from the fact that the righthand side of~\eqref{eq:intro_simplified} is not the entire space $S^d(\V)$, by a simple dimension count. In the other cases, we let $1\leq \ell\leq d-1$ be the largest integer for which $i_1=\dots=i_{\ell}$, and apply the induction hypothesis in degree $d-\ell$ by showing that it suffices to establish a similar non-containment as~\eqref{eq:intro_simplified} for the tensor $P_{\V,(d-\ell)}^{\vee}(v_{i_{\ell+1}}\otimes \dots \otimes v_{i_d})$. To make this reduction, we contract along $v_{i_1}^{\otimes \ell}$ while ensuring that this does not affect the dimension of the righthand side too much. This requires a clever choice of ordering $\succ$, along with a lower bound on the dimension of a generic contraction of a linear subspace (Lemma~\ref{lemma:generic_hook:new}). After this reduction, the desired statement ``looks like" the $s=0$ case in degree $(d-\ell)$, since by the non-degeneracy assumption the variety $\X$ ``looks like" the entire space $\V$ in degree $(d-\ell)$, in the sense of~\eqref{eq:nondegenerate}. This allows us to apply the induction hypothesis, completing the proof.

\paragraph{Comparison to the FOOBI algorithm}
In earlier work of De Lathauwer, Cardoso and Castaing, an algorithm (often referred to as the \textit{FOOBI algorithm}) is proposed for finding the rank-one matrices in a linear subspace of matrices~\cite{de2006link,DLCC}. Our algorithm specializes to theirs in the case when $\X=\X_1$ is the set of rank-one matrices, as we now briefly describe. In the FOOBI algorithm, a linear map $\Phi_{\X_1}$ is constructed with $\ker(\Phi_{\X_1})=I(\X_1)_2^{\perp}$, where $I(\X_1)_2 \subseteq S^2(\V^*)$ is the span of the $2 \times 2$ determinants, which cut out $\X_1$. Their algorithm is then essentially the same as ours in this specialized setting, with the map $\Phi_{\X_1}$ serving as a proxy for $I(\X_1)_2^{\perp}$.
We have found that working directly with the subspace $I(\X_1)_2^{\perp}$ greatly simplifies most arguments and notation.

The papers~\cite{de2006link,DLCC} also claim a similar genericity guarantee as we do in the special case $\X=\X_1$ (see Lemma 2.3 of \cite{de2006link}).
However, we show that this lemma is false by presenting an explicit counterexample in Appendix~\ref{app:counterexample}, and also identify the incorrect step in their proof. Our analysis thus provides a correct genericity guarantee that holds in a much more general setting. See Section~\ref{sec:tensor:related} for further discussion and other related work.



%

\subsection{Related work on tensor decompositions} \label{sec:tensor:related} 


There is a rich body of work on low-rank tensor decompositions 
where the goal is to express a given tensor as a sum of rank-$1$ tensors. 
Considering the intractability of the tensor decomposition problem~\cite{Has90, HL}, several different assumptions on the input tensor have been made to overcome the worst-case intractability. We focus on algorithms that run in polynomial time and provably recover the rank-$1$ components (this also implies uniqueness) for \emph{generically} chosen tensors. See \cite{Vij20} for references to other related lines of work.    

The first algorithm for tensor decompositions was the simultaneous diagonalization method~\cite{Har72, LRA93}, which was used to recover the decomposition for {\em generically chosen} tensors in $\field^{n \times n \times n}$ of rank $R \le n$.\footnote{This is also sometimes called Jennrich's algorithm, named after Robert Jennrich, who Harshman credits for the first uniqueness result for tensor decompositions~\cite{harshman1970foundations}. Harshman gave an alternate proof of uniqueness using the simultaneous diagonalization method (see the Theorem on page 2 of \cite{Har72}).
} We use this algorithm as a subroutine in our algorithm; see Section~\ref{sec:jennrich} for details. We are not aware of any polynomial time guarantee for generically chosen third-order tensors in $\field^{n \times n \times n}$ of rank $R> (1+\varepsilon) n$ for constant $\varepsilon>0$; see \cite{BCMVopen} for a related open question.\footnote{Some existing algorithms have a running time dependence of $n^{O(t)}$ to handle generic instances of rank $n+t$ ~\cite{domanov2017canonical,ChenRademacher20}.} 

Our algorithm is a broad generalization of an algorithm proposed in~\cite{de2006link} for third order tensor decompositions. This work also claims a similar genericity guarantee to our Corollary~\ref{cor:intro:tensor_decomp}, but we show that a crucial lemma in their proof is incorrect (see Appendix~\ref{app:counterexample}). Unfortunately, this error percolates to genericity guarantees in subsequent works, most notably that of the well-known FOOBI algorithm for symmetric fourth-order tensor decompositions~\cite{DLCC}. Our genericity guarantee (Theorem~\ref{thm:intro:intersection}) provides a correct proof of these results, uses a completely different proof technique, and holds in a much more general setting that applies to other types of decomposition problems and entangled subspace problems. We emphasize that, while the genericity guarantees in these works appear to be incorrect, the algorithms and computational methods presented are correct to our knowledge.

Several other generalizations of tensor decompositions that have been studied previously are also captured by $(\X, \W)$-decompositions.
Some sufficient conditions for generic (non-algorithmic) uniqueness results were explored in \cite{domanov2016generic}. When $\X=\X_r$ is the variety of rank $r$ matrices (of a given dimension), $(\X_r, \W)$-decompositions correspond to $r$\textit{-aided decompositions} (also called $(r,r,1)$\textit{-block decompositions} and \textit{max ML-$(r,r,1)$ decompositions}). Such $r$-aided decompositions have applications in signal processing and machine learning, among others~\cite{kolda2009tensor,comon2010handbook,cichocki2015tensor,sidiropoulos2017tensor}, and were also studied, for example, in~\cite{de2008decompositionspart1,de2008decompositionspart2,de2008decompositionspart3,domanov2020uniqueness}. Our general result in Theorem~\ref{thm:decomp} gives guarantees for $r$-aided decompositions, as described in Corollary~\ref{cor:block_decomp}. We are unaware of such polynomial time genericity guarantees prior to our work. 

In other related work, there also exist algorithmic guarantees for tensor decompositions with random components that can handle larger rank (e.g., random tensors in $\real^{n \times n \times n}$ of rank $\widetilde{O}(n^{3/2})$ ~\cite{GM15}). However, these make strong assumptions about the components like incoherence (near orthogonality), which are not satisfied by \emph{generic} instances.  
There also exists a line of work on smoothed analysis guarantees~\cite{BCMV, MSS, BCPV} that are similar in flavor to genericity guarantees, but provide robust guarantees for tensor decompositions under slightly stronger assumptions. Obtaining smoothed analysis analogs of our results is an interesting open question.



Finally, tensor decompositions have seen a remarkable range of applications for algorithmic problems in data science and machine learning, including parameter estimation of latent variable models like mixtures of Gaussians, hidden Markov models, and even for learning shallow neural networks~\cite{Moitrabook, Anandkumarbook}.
Our work shows strong uniqueness results and efficient polynomial time algorithms for a broader class of low-rank decomposition problems, and may present a powerful algorithmic toolkit for applications in these domains.

\paragraph{Outline}

In Section~\ref{sec:prelim} we introduce some notation, mathematical preliminaries and some existing algorithmic subroutines that will be used in later sections. Section~\ref{sec:algorithm} describes the algorithm and shows some correctness properties of the algorithm. Section~\ref{sec:simplified} proves Theorem~\ref{thm:intro:intersection} (see Corollary~\ref{cor:generic_simplified}  for the formal claim and proof). The applications to quantum entanglement are presented in Section~\ref{sec:entanglement}, while the applications to low-rank decompositions are presented in Section~\ref{sec:tensors}.

\section{Mathematical preliminaries} \label{sec:prelim}
In this section we review some mathematical preliminaries for this paper. We begin with some miscellaneous definitions, and then review the symmetric subspace, basic notions from algebraic geometry, some relevant examples of varieties, decompositions over varieties, and the simultaneous diagonalization algorithm.

Let $[R]=\{1,\dots, R\}$ when $R$ is a positive integer. For a finite, ordered set $S$ and a positive integer $d$, let $S^{\times d}$ be the $d$-fold cartesian product of $S$, and let
\begin{align}
S^{\vee d} = \{(a_1,\dots, a_{d}): a_1,\dots, a_{d} \in S\quad \text{and}\quad a_1 \leq \dots \leq a_d\}.
\end{align}
For example, if $S=[R]$, then
\begin{align}
[R]^{\vee d} = \{(a_1,\dots, a_{d}): 1\leq a_1 \leq \dots \leq a_d\leq R \}.
\end{align}

Throughout this work, we let $\field$ denote either the real or complex field. For an $\field$-vector space $\V$ of dimension $n$, let $\{e_1,\dots, e_n\}$ be a standard basis for $\V$, and let $\{x_1,\dots, x_n\}$ be the dual basis for $\V^*$.

\subsection{The symmetric subspace}
Let $\V$ be an $\field$-vector space of dimension $n$. For a positive integer $d$, let $\field[x_1,\dots, x_n]_d$ be the vector space of homogeneous degree-$d$ polynomials on $\V$ (in addition to the zero polynomial), and let $\field[x_1,\dots, x_n]=\bigoplus_{d=0}^\infty \field[x_1,\dots,x_n]_d$ be the polynomial ring on $\V$. Let $\mathfrak{S}_d$ be the group of permutations of $d$ elements, and let $S^d(\V)$ be the $d$-th symmetric power of $\V$, which can be identified (since $\setft{Char}(\field)=0$) with the set of tensors $T\in \V^{\otimes d}$ that are invariant under the action of $\mathfrak{S}_d$ on $\V^{\otimes d}$ which permutes the copies of $\V$. As mentioned above, we frequently invoke the canonical vector space isomorphism $\field[x_1,\dots, x_n]_d \cong S^d(\V^*)$.

Let $P^{\vee}_{\V,d}: \V^{\otimes d} \rightarrow S^d(\V)$ be the projection map. The standard basis $\{e_1,\dots, e_n\}$ of $\V$ induces a basis of $S^d(\V)$ given by
\ba
\{P^{\vee}_{\V,d}(e_{a_1}\otimes \dots \otimes e_{a_d}) : a \in [n]^{\vee d}\}.
\ea

\paragraph{Contraction or Hook:} For $\field$-vector spaces $\V_1,\dots, \V_m$, an index $i \in [d]$, a vector $v \in \V_i^*$, and a tensor $T \in \V_1 \otimes \dots \otimes \V_d$, we define the \textit{contraction} of $T$ with $v$ in the $i$-th mode, denoted $v \hook_i T$, to be the tensor obtained by regarding $T$ as a map $\V_i^* \rightarrow \V_1 \otimes \dots \otimes \V_{i-1}\otimes \V_{i+1} \otimes \dots \otimes \V_d$ and evaluating at $v$:
\ba
v \hook_i T := T(v) \in \V_1\otimes \dots \otimes \V_{i-1}\otimes \V_{i+1}\otimes \dots \otimes \V_d.
\ea


\subsection{Algebraic geometry}\label{sec:MP:AG}

A \textit{algebraic set} (or an \textit{algebraic variety}, or simply a \textit{variety}) in $\V$ is a subset $\X \subseteq \V$ for which there exists a set of polynomials $f_1,\dots, f_p \in \field[x_1,\dots, x_n]$ such that
\ba
\X=\{v \in \V : f_1(v)=\dots=f_p(v)=0\}.
\ea
In this case, we say that $\X$ is \textit{cut out} by $f_1,\dots, f_p$. We say that a variety $\X$ is \textit{conic} if $\field \X=\X$. It is straightforward to verify that a variety $\X$ is conic if and only if it is cut out by homogeneous polynomials, which can furthermore be chosen to all have the same degree $d$. The \textit{Zariski topology} is the topology on $\V$ with closed sets given by the varieties in $\V$. We therefore also refer to a variety as a \textit{Zariski closed} (or simply, a \textit{closed}) subset of $\V$. A subset of $\V$ is called \textit{locally closed} if it is the intersection of an open subset of $\V$ with a closed subset of $\V$. A subset of $\V$ is called \textit{constructible} if it is a finite union of locally closed subsets of $\V$.
A subset $\A \subseteq \V$ is called \textit{irreducible} if it cannot be written as a finite union of proper closed subsets of $\A$ (with respect to the subspace topology on $\A$). Any Zariski closed subset $\X \subseteq \V$ can be written (uniquely, up to reordering terms) as a finite union of irreducible varieties $\X=\X_1 \cup \dots \cup \X_k$. The irreducible varieties $\X_1,\dots, \X_k \subseteq \V$ are called the \textit{irreducible components} of $\X$.

Let $\X\subseteq \V$ be a conic, irreducible variety. We say that $\X$ is \textit{non-degenerate} if it is not contained in any proper linear subspace of $\V$, i.e. $\spn(\X)=\V$. More generally, we say that $\X \subseteq \V$ is \textit{non-degenerate of order $\tilde{d}$} if there does not exist any homogeneous degree-$\tilde{d}$ polynomials that vanish on $\X$, i.e. if the set
\ba
I(\X)_{\tilde{d}} :=\{f \in \field[x_1,\dots, x_n]_{\tilde{d}} : f(v)=0 \quad \text{for   all} \quad v \in \X\}
\ea
is equal to $\{0\}$. More generally, we will say that a reducible variety $\X$ is non-degenerate of order $\tilde{d}$ if all of its irreducible components are non-degenerate of order $\tilde{d}$. For the purpose of inductive arguments, we adopt the convention that every variety is non-degenerate of order zero. The set $I(\X)_d$ is called the \textit{degree-$d$-component of the ideal of $\X$}. The set $I(\X):= \oplus_{d=0}^{\infty} I(\X)_d$ is called the \textit{ideal of} $\X$.
Viewing the elements of $\field[x_1,\dots, x_n]_d$ as elements of $S^d(\V^*)$, we have
\ba\label{eq:ideal}
I(\X)_d^{\perp}=\spn\{v^{\otimes d} : v \in \X\}\subseteq S^d(\V).
\ea
As a consequence, we see that an irreducible, conic variety $\X$ is non-degenerate of degree $\tilde{d}$ if and only if
\ba
\spn\{v^{\otimes \tilde{d}} : v \in \X\}=S^{\tilde{d}}(\V).
\ea
Note that if $\X$ is cut out in degree $d$, then $\X$ is cut out by $p=\dim(I(\X)_d)$ many linearly independent homogeneous polynomials of degree $d$, where $\dim(I(\X)_d)$ denotes the dimension of $I(\X)_d$ as an $\field$-vector space.

\paragraph{Genericity:} For a variety $\X \subseteq \V$, we say that a property holds for a \textit{generically chosen} element $v \in \X$ if there exists a Zariski open dense subset (in the induced topology on $\X$) $\A \subseteq \X$ such that the property holds for all $v \in \A$. Zariski open dense sets are massive: In particular, Zariski open dense subsets of $\V$ are full measure with respect to any absolutely continuous measure, and Zariski open dense subsets of a variety $\X$ are dense in $\X$ in the Euclidean topology. For varieties $\X_1,\dots, \X_{R} \subseteq \V$, the cartesian product $\X_1\times \dots \times \X_{R} \subseteq \V^{\times R}$ is again a variety, and we say that a property holds for \textit{generically chosen} elements $v_1 \in \X_1,\dots, v_{R} \in \X_{R}$ if there exists a Zariski open dense subset $\A \subseteq \X_1 \times \dots \times \X_{R}$ for which the property holds for all $(v_1,\dots, v_{R})\in \A$ (i.e., if it holds for a generically chosen element $v \in \X_1 \times \dots \times \X_{R}$).

\paragraph{Genericity over $\real$ and $\complex$:} We will be proving and using genericity results over $\real$ and $\complex$ simultaneously. To this end, we present a basic fact which will allow us to translate genericity results over $\complex$ to genericity results over $\real$. Let $\setft{Cl}_Z^{\field}(\cdot)$ denote the Zariski closure over $\field$.

\begin{fact}\label{fact:real_vs_complex}
Let $\X \subseteq \real^n\subseteq \complex^n$ be a real variety, let $\T=\setft{Cl}_Z^{\complex}(\X)$ be its complex Zariski closure, and let $\A \subseteq \T$ be a Zariski open dense subset. Then the following two properties hold:
\begin{enumerate}
\item $\A \cap \X$ is Zariski open in $\X$ over $\real$
\item $\A \cap \X$ is Zariski dense in $\X$ over $\real$.
\end{enumerate}
\end{fact}
\begin{proof}
The first property follows from the fact that $\A \cap \X=\A \cap \real^n \cap \X$ by construction, and $\A \cap \real^n\subseteq \real^n$ is Zariski open.\footnote{The real part of a Zariski open set is Zariski open over $\real$. Indeed, $\A$ is the complement of some Zariski closed set $\X=\{v \in \complex^n :f_1(v)=\dots=f_{p}(v)=0\}\subseteq \complex^n,$ so $\A\cap \real^n$ is the complement of the Zariski closed subset of $\real^n$ cut out by the $2 p$ polynomials formed by taking the real and imaginary parts of each $f_i$.}

For the second property, suppose toward contradiction that there exists a real variety $\Z\subseteq \real^n$ for which
\ba
\A \cap \X \subseteq \Z \subsetneq \X.
\ea
Let $\U=\setft{Cl}_Z^{\complex}(\Z) \subseteq \T$ be the complex Zariski closure of $\Z$, and note that $\U \cap \X=\Z$ (this follows from the fact that $\U \cap \real^n = \Z$). This gives
\ba
\A \cap \X \subseteq \U \cap \X \subsetneq \X.
\ea

But this implies that $\X \subseteq \U \cup (\T\setminus \A) \subseteq \T$. Since $\X\subseteq \T$ is Zariski dense, and $ \U \cup (\T\setminus \A) \subseteq \T$ is Zariski closed, it follows that $ \U \cup (\T\setminus \A) = \T$, so $\A \subseteq \U \subsetneq \T$. This is a contradiction to $\A \subseteq \T$ being Zariski-dense, and completes the proof that $\A \cap \X$ is Zariski dense in $\X$ over $\real$.
\end{proof}

\subsection{Examples of varieties}\label{sec:XR}

In this section, we introduce several well known examples of conic varieties, which we will use in later sections to demonstrate applications of our algorithm. These include determinantal varieties of matrices, the variety of product tensors, the variety of biseparable tensors, and the variety of slice rank one tensors.

Let $n_1$, $n_2$, and $r\leq \min\{n_1,n_2\}$ be positive integers, let $\V_1$ and $\V_2$ be $\field$-vector spaces of dimensions $n_1$ and $n_2$, and let
\ba\label{eq:XR}
\X_r := \{v \in \V_1 \otimes \V_2 : \rank(v) \leq r\} \subseteq \V_1 \otimes \V_2,
\ea
where $\rank(v)$ denotes the rank of $v \in \V_2 \otimes \V_2$, viewed as an $n_2 \times n_1$ matrix. More precisely, this is the rank of $v$ when $v$ is viewed as an element of $\setft{Hom}_{\field}(\V_2^*,\V_1)$ under the isomorphism
\ba\label{eq:tensor_iso}
\V_1 \otimes \V_2 \cong \setft{Hom}_{\field}(\V_2^*,\V_1),
\ea
where $\setft{Hom}_{\field}(\V_2^*,\V_1)$ denotes the set of $\field$-linear maps from $\V_2^*$ to $\V_1$. This map sends $v_1 \otimes v_2 \in \V_1 \otimes \V_2$ to the map $f \mapsto f(v_2) v_1$, and extends linearly. In coordinates, this is simply the map which regards a tensor of dimension $\dim(\V_1) \dim(\V_2)$ as a $\dim(\V_1)\times \dim(\V_2)$ matrix.

We will sometimes use the notation $\X_r^{\real}$ and $\X_r^{\complex}$ to emphasize the field.  It is a standard fact that $\X_r$ is a conic variety cut out by the $(r+1) \times (r+1)$ minors (these minors have degree ${r+1}$, and there are $\binom{n_1}{r+1}\binom{n_2}{r+1}$ of them). A slightly less standard fact is that $\X_r$ has no equations in degree $r$. Over $\complex$, this follows from the fact that the $(r+1)\times (r+1)$ minors generate the ideal of $\X_r^{\complex}$~\cite{harris2013algebraic}. Over $\real$, it follows from e.g.~\cite[Theorem 2.2.9.2]{mangolte2020real} that the Zariski closure of $\X_r^{\real}$ in $\complex^n$ is $\X_r^{\complex}$. By Hilbert's Nullstellensatz, it follows that any real polynomial which vanishes on $\X_r^{\real}$ must also vanish on $\X_r^{\complex}$. Thus, the $(r+1)\times (r+1)$ minors also generate the ideal of $\X_r^{\real}$, so in particular, $\X_r^{\real}$ has no equations in degree $r$.

Our further examples will be subsets of tensor product spaces with more factors: Let $n_1,\dots, n_m$ be positive integers, let $\V_1,\dots, \V_m$ be $\field$-vector spaces of dimensions $n_1,\dots, n_m$, and let ${\V=\V_1\otimes \dots \otimes \V_m}$. Let
\ba
\X_{\setft{Sep}}=\left\{v_1 \otimes \dots \otimes v_m :  v_1 \in \V_1,\dots, v_m \in \V_m\right\}
\ea
be the set of~\textit{product tensors} (or~\textit{separable tensors}). Then $\X_{\setft{Sep}}$ is non-degenerate and is cut out by exactly
\ba\label{eq:Sep_p}
p=\binom{n_1 \cdots n_m +1}{2}-\binom{n_1+1}{2}\cdots \binom{n_m+1}{2}
\ea
many linearly independent homogeneous polynomials of degree $d=2$. Indeed, it is well-known that $\X_{\setft{Sep}}$ is non-degenerate and cut out by degree $d=2$ polynomials~\cite{harris2013algebraic}. The number follows from the fact that $p=\dim(I(\X_{\setft{Sep}})_2)$ (see Section~\ref{sec:MP:AG}), equation~\eqref{eq:ideal}, and the fact that
\ba
\spn\{v^{\otimes 2} : v \in \X_{\setft{Sep}}\}=S^2(\V_1)\otimes \dots \otimes S^2(\V_m),
\ea
which has dimension $\binom{n_1+1}{2}\cdots \binom{n_m+1}{2}$. Let
\ba
\X_B=\bigcup_{\substack{T \subseteq [m]\\ 1\leq \abs{T} \leq \floor{m/2}}}\left\{v\in \V : \rank\left(v : \bigotimes_{i \in T} \V_i^* \rightarrow \bigotimes_{j \in [m]\setminus T} \V_j\right) \leq 1\right\}
\ea
be the set of~\textit{biseparable tensors}. Then this is the decomposition of $\X_B$ into irreducible components, and the irreducible component indexed by $T\subseteq [m]$ is non-degenerate and cut out by
\ba
p_T = \binom{\prod_{i \in T} n_i}{2}\binom{\prod_{j \in [m]\setminus T} n_j}{2}
\ea
many linearly independent homogeneous polynomials of degree $d=2$ (this follows directly from the analogous statement for $\X_1$ above). Similarly, let
\ba
\X_S=\bigcup_{i \in [m]}\left\{v\in \V : \rank\left(v : \V_i^* \rightarrow \bigotimes_{j \in [m]\setminus \{i\}} \V_j\right) \leq 1\right\}
\ea
be the set of~\textit{slice rank 1 tensors}. Then $\X_S$ is non-degenerate, this is the decomposition of $\X_S$ into irreducible components, and the component indexed by $i\in [m]$ is non-degenerate and cut out by
\ba
p_i = \binom{n_i}{2}\binom{\prod_{j \in [m]\setminus \{i\}} n_j}{2}
\ea
many linearly independent homogeneous polynomials of degree $d=2$ (this again follows directly from the analogous statement for $\X_1$ above). We will also consider the set of symmetric product tensors. If $\V$ is an $\field$-vector space of dimension $n$, then we define
\ba
\X_{\setft{Sep}}^{\vee} = \X_{\setft{Sep}} \cap S^m(\V) = \{ \alpha v^{\otimes m} : \alpha \in \field, v \in \V\}\subseteq S^m(\V)
\ea
to be the set of \textit{symmetric product tensors} in $S^m(\V)$. The set $\X_{\setft{Sep}}^{\vee}$ forms a non-degenerate algebraic variety that is cut out by
\ba
p=\binom{\binom{n+m-1}{m} +1}{2} - \binom{n+2m+1}{2m}
\ea
many linearly independent homogeneous polynomials of degree $d=2$ (this calculation is similar to the analogous calculation for $\X_{\setft{Sep}}$).

\subsection{Decompositions over varieties}\label{sec:variety_decomp}

For an $\field$-vector space $\T$ of dimension $n$, a conic, non-degenerate variety $\Y \subseteq \T$, and a vector $T \in \T$, a \textit{$\Y$-decomposition} of $T$ is a set $\{v_1,\dots, v_R\} \subseteq \Y$ for which
\ba\label{eq:Y_decomp}
T=\sum_{a \in [R]} v_a.
\ea
The number $R$ is called the \textit{length}, or~\textit{rank} of this decomposition. The \textit{$\Y$-rank} of $T$ is the minimum length of any $\Y$-decomposition of $T$. We say that a $\Y$-rank decomposition $\{v_1,\dots, v_R\} \subseteq \Y$ is the \textit{unique $\Y$-rank decomposition} of $T$ if every other decomposition of $T$ has length greater than $R$. We will sometimes abuse terminology and refer to an expression of the form~\eqref{eq:Y_decomp} as a $\Y$-decomposition.

In this work, we study a particular type of $\Y$-decomposition. For $\field$-vector spaces $\V$ and $\W$ and a conic, non-degenerate variety $\X \subseteq \V$, we study \textit{$(\X,\W)$-decompositions} (also called~\textit{simultaneous $\X$-decompositions}): these are $\Y$-decompositions under the choice
\ba\label{eq:Y}
\Y=\{v \otimes w : v \in \X \quad \text{and} \quad w \in \W\}\subseteq \V\otimes \W.
\ea
For example, when $\V=\V_1\otimes \V_2$ is a tensor product space and $\X_1$ is the determinantal variety introduced in Section~\ref{sec:XR}, $(\X_1,\W)$-decompositions exactly correspond to tensor decompositions, i.e. expressions of a tensor $T \in \V_1 \otimes \V_2 \otimes \W$ as a sum of terms of the form $v_1 \otimes v_2 \otimes w$. More generally, $(\X_r, \W)$-decompositions correspond to $r$-\textit{aided rank decompositions} (also called \textit{max ML rank-$(r,r,1)$ decompositions}, and \textit{$(r,r,1)$-block decompositions}). Aided decompositions have applications in signal processing and machine learning, among others~\cite{kolda2009tensor,comon2010handbook,cichocki2015tensor,sidiropoulos2017tensor} and were also studied, for example, in~\cite{de2008decompositionspart1,de2008decompositionspart2,de2008decompositionspart3,domanov2020uniqueness}. As one more example (which also generalizes $(\X_1,\W)$-decompositons), when $\V=\V_1\otimes \cdots \otimes \V_{m}$ and $\X=\X_{\setft{Sep}}\subseteq \V$ is the set of product tensors, $(\X_{\setft{Sep}},\W)$-decompositions correspond to tensor decompositions in $\V_1\otimes \dots \otimes \V_m \otimes \W$, i.e. expressions of a tensor $T \in \V_1\otimes \dots \otimes \V_m \otimes \W$ as a sum of terms of the form $v_1\otimes \dots \otimes v_m \otimes w$.


%
%

We will say that a property holds for a~\textit{generically chosen} element $T \in \V \otimes \W$ of $(\X,\W)$-rank at most $R$ if there exists a Zariski open dense subset $\A\subseteq \X^{\times R}\times \W^{\times R}$ such that for all $(v_1,\dots, v_R, w_1,\dots, w_R)\in \A$, the property holds for $T=\sum_{a=1}^R v_a \otimes w_a$.

\subsection{Simultaneous diagonalization algorithm}\label{sec:jennrich}

In this section, we review the simultaneous diagonalization algorithm~\cite{Har72} (that is sometimes referred to as Jennrich's algorithm or Harshman's algorithm), which we will use as a subroutine in our algorithm. For $\field$-vector spaces $\V$ and $\W$, we recall the natural isomorphism
\ba
\V \otimes \W \cong \setft{Hom}_{\field}(\W^*,\V),
\ea
(see~\eqref{eq:tensor_iso}). We will invoke this isomorphism several times in the simultaneous diagonalization algorithm and throughout this paper. For example, we will view a tensor $T \in \V_1 \otimes \V_2\otimes \V_3$ as an element of $\setft{Hom}_{\field}(\V_1^*, \V_2\otimes \V_3)$, and also as an element of $\setft{Hom}_{\field}((\V_2 \otimes \V_3)^*,\V_1).$ For a linear map $X \in \setft{Hom}_{\field}(\V_2^*,\V_3)$, let $X^+ \in \setft{Hom}_{\field}(\V_3,\V_2^*)$ be the Moore-Penrose pseudoinverse of $X$. 

\begin{figure}[htbp]
\begin{center}
\fbox{\parbox{0.98\textwidth}{

\begin{center}\textbf{\Large{Simultaneous diagonalization algorithm}}\end{center}
\textbf{Input:} A tensor $T \in \V_1 \otimes \V_2 \otimes \V_3$.
\begin{enumerate}
\item Choose $f, g \in (\field^{n_1})^*$ uniformly at random (according to e.g. the uniform spherical measure).
\item Let $R=\rank(T(f) T(g)^{+}).$ Compute the eigenvalues and eigenvectors of $T(f) T(g)^{+}$. If there are repeated non-zero eigenvalues, output: ``Fail.'' Otherwise, let $\{\lambda_1,\dots, \lambda_R\}$ be the non-zero eigenvalues of $T(f) T(g)^{+}$, and let $\{v_1,\dots, v_R\}$ be the (unique, up to scale) corresponding eigenvectors.
\item Compute the eigenvalues and eigenvectors of $T(f)^+ T(g)$. If the non-zero eigenvalues are not $\{\lambda_1^{-1},\dots,\lambda_R^{-1}\},$ then output: ``Fail." Otherwise, let $\{w_1,\dots, w_R\}$ be the corresponding eigenvectors.
\item Let $\{h_{i} : i \in [R]\} \subseteq (\field^{n_2} \otimes \field^{n_3})^*$ be any set of linear functionals that is dual to $\{v_a \otimes w_a : a \in [R]\}$, i.e. for which $h_a(v_b \otimes w_b) = \delta_{a,b}$ for all $a,b \in [R]$. Let $u_a=T(h_a) \in \field^{n_1}$ for all $a \in [R]$, viewing $T$ as a linear map $(\field^{n_2} \otimes \field^{n_3})^* \rightarrow \field^{n_1}$. If $u_a \in \spn\{u_b\}$ for some $a \neq b \in [R]$, then output: ``Fail.'' Otherwise, output: ``$\{u_a \otimes v_a \otimes w_a : a \in [R]\}$ is the unique tensor rank decomposition of $T$.''
\end{enumerate}
}
}
\end{center}
\end{figure}

\begin{fact}[Correctness of the simultaneous diagonalization algorithm]\label{fact:jennrich}
Let $T \in \field^{n_1}\otimes \field^{n_2} \otimes \field^{n_3}$ be a tensor admitting a decomposition of the form $\{u_a \otimes v_a \otimes w_a : a \in [R]\}$, where (i) $\{v_1,\dots, v_R\}$ is linearly independent, (ii) $\{w_1,\dots, w_R\}$ is linearly independent, and (iii) $u_a \notin \spn\{u_b\}$ for all $a \neq b \in [R]$ i.e., $\{u_1, \dots, u_R\}$ has Kruskal rank at least $2$. Then this is the unique tensor rank decomposition of $T$, and with probability 1 over the choice of $f,g \in (\field^{n_1})^*$ in Step 1, the simultaneous diagonalization algorithm outputs ``$\{u_a \otimes v_a \otimes w_a : a \in [R]\}$ is the unique tensor rank decomposition of $T$."
\end{fact}
In particular, Fact~\ref{fact:jennrich} shows that for any tensor $T \in \field^{n_1}\otimes \field^{n_2} \otimes \field^{n_3}$ admitting a decomposition of the form $\{u_a \otimes v_a \otimes w_a : a \in [R]\}$, where $\{u_1,\dots, u_R\}, \{v_1,\dots,v_R\},$ and $\{w_1,\dots,w_R\}$ are all linearly independent, this is the unique tensor rank decomposition of $T$, and it is computed by the simultaneous diagonalization algorithm. It also shows that, when $n_1 \geq 2$, the simultaneous diagonalization algorithm computes the (unique) tensor rank decomposition of generically chosen tensors in $\field^{n_1}\otimes \field^{n_2} \otimes \field^{n_3}$ of tensor rank at most $\min\{n_2,n_3\}$.
\begin{proof}[Proof of Fact~\ref{fact:jennrich}]
The fact that $\{u_a \otimes v_a \otimes w_a : a \in [R]\}$ is the unique tensor rank decomposition of $T$ follows from Jennrich's theorem~\cite{harshman1970foundations, Har72} (or more generally, Kruskal's theorem, see~\cite{kruskal1977three} or~\cite{tensor}). If $T$ admits such a decomposition, then
\ba
T=\sum_{a\in [R]} u_a \otimes v_a \otimes w_a,
\ea
so the eigenvalues of $T(f) T(g)^{+}$ are
\ba
\Big\{\frac{f(u_a)}{g(u_a)} : a \in [R]\Big\},
\ea
which are clearly distinct for generically chosen $f,g \in (\field^{n_1})^*$, since $u_a \notin \spn\{u_b\}$ for all $a \neq b \in [R]$. The corresponding eigenvectors are $\{v_1,\dots, v_R\}$. Similarly, the eigenvalues of $T(f)^+ T(g)$ are the reciprocals:
\ba
\Big\{\frac{g(u_a)}{f(u_a)} : a \in [R]\Big\},
\ea
with corresponding eigenvectors $\{w_1,\dots, w_R\}$. It is also clear that $u_a=T(h_a)$, so the simultaneous diagonalization algorithm outputs ``$\{u_a \otimes v_a \otimes w_a : a \in [R]\}$ is the unique tensor rank decomposition of $T$." This completes the proof.
\end{proof}

\section{The algorithm for computing $\U \cap \X$}\label{sec:algorithm}

Suppose we are handed a basis $\{u_1,\dots, u_{R}\}$ for an $R$-dimensional linear subspace $\U \subseteq \V$, and we wish to describe the intersection of $\U$ with a conic variety $\X$. In this section, we propose an algorithm that (if it does not output ``Fail"), either certifies $\U \cap \X=\{0\}$ (in which case we will say that $\U$ \textit{trivially intersects} $\X$), or else finds all the elements of $\U \cap \X$, provided that there are at most $R$ of them up to scalar multiples. Later on, in Section~\ref{sec:simplified} we prove that the algorithm does not return ``Fail'' almost surely under the conditions of Theorem~\ref{thm:intro:intersection}.

Since $\X$ is a conic variety, there exists a positive integer $d$ and a finite set of homogeneous degree-$d$ polynomials $f_1,\dots, f_p \in S^d(\V^*)$ that cut out $\X$.
Let $I=\langle f_1,\dots, f_p \rangle$ be the ideal generated by $f_1,\dots, f_p$, and $I_d=\spn\{f_1,\dots, f_p\}\subseteq S^d(\V^*)$.

Correctness of our algorithm relies on the following two observations: Observation~\ref{obs:cert}, a sufficient condition for $\U$ to trivially intersect $\X$; and Observation~\ref{obs:find}, a sufficient condition for there to be only $s \leq R$ elements of $\U \cap \X$, up to scalar multiples.
\begin{obs}\label{obs:cert}
If $S^d(\U) \cap I_d^{\perp}=\{0\}$, then $\U \cap \X=\{0\}$.
\end{obs}
\begin{proof}
The observation is immediate from the fact that for any vector $u \in \U \cap \X$, it holds that $u^{\otimes d} \in S^d(\U) \cap I(\X)_d^{\perp}$.
\end{proof}

\bnote{Remark added:}
\begin{remark}[Relation to Hilbert's projective nullstellensatz over $\complex$]\label{rmk:nullstellensatz}
By Hilbert's projective nullstellensatz, if $\field=\complex$ then $\U \cap \X=\{0\}$ if and only if there exists a positive integer $D$ for which $I(\U)_d+I_d = S^D(\V^*)$ (and furthermore, $D$ can be chosen less than $d^{O(n)}$; see e.g. \cite{kollar1988sharp}). If $\field=\real$, then $I(\U)_d+I_d = S^D(\V^*)$ still implies $\U \cap \X=\{0\}$, but the reverse implication no longer holds in general. Dualizing, note that $I(\U)_d+I_d = S^D(\V^*)$ if and only if $S^d(\U) \cap I_d^{\perp}=\{0\}$, so this is just a degree-$d$ \textit{Nullstellensatz certificate} for checking whether $\U\cap \X=\{0\}$.
\end{remark}


\begin{obs}\label{obs:find}
If $d \geq 2$ and there exists a set of linearly independent vectors $v_1,\dots, v_R \in \V$ for which $S^d(\U) \cap I_d^{\perp}=\spn\{v_1^{\otimes d},\dots, v_s^{\otimes d}\},$ then the only elements of $\U \cap \X$ are $v_1,\dots, v_s$ (up to scalar multiples).
\end{obs}
\begin{proof}
By, for example,~\cite[Theorem 3.2]{1751-8121-48-4-045303} or~\cite[Corollary 19]{tensor}, it follows from linear independence of $v_1,\dots, v_R$ that $v_1^{\otimes d}, \dots, v_s^{\otimes d}$ are the only symmetric product tensors in $S^d(\U) \cap I_d^{\perp}$ up to scale, so $v_1,\dots, v_s$ are the only elements of $\U \cap \X$ up to scale.
\end{proof}

This inspires Algorithm~1 for computing the intersection $\U\cap \X$.

\begin{figure}[htbp]
\begin{center}
\fbox{\parbox{0.98\textwidth}{

\begin{center}\textbf{\Large{Algorithm 1: Computing $\U\cap \X$.}}\end{center}
\textbf{Input:} A basis $\{u_1,\dots, u_{R}\}$ for a linear subspace $\U \subseteq \V$, and a collection of homogeneous degree-$d$ polynomials $f_1,\dots, f_p$ that cut out a conic variety $\X\subseteq \V$.
\begin{enumerate}
\item Let $I_d=\spn\{f_1,\dots, f_p\} \subseteq S^d(\V^*)$, and compute a basis $\{P_1,\dots, P_s\} \subseteq S^d(\V)$ for $S^d(\U) \cap I_d^{\perp}$. If $S^d(\U) \cap I_d^{\perp}=0$, then $\U \cap \X=\{0\}$. Output: ``$\U$ trivially intersects $\X$."
\item If $s > R$, then output: ``Fail." Otherwise, construct the tensor
\ba
T=\sum_{i=1}^s e_i \otimes P_i \in \field^{s} \otimes S^d(\V).
\ea
Regarding $T$ as a 3-mode tensor
\ba
T \in \field^s \otimes \V \otimes (\V)^{\otimes d-1},
\ea
run the simultaneous diagonalization algorithm on $T$. If the simultaneous diagonalization algorithm outputs a decomposition of $T$ of the form $\{z_i \otimes v_i^{\otimes d} : i \in [s]\}$ for some linearly independent $z_1,\dots, z_s \in \field^{s}$ and linearly independent $v_1,\dots, v_s \in \V$, then output: ``The only elements of $\U \cap \X$ are $\{v_1,\dots, v_s\}$ (up to scale)." Otherwise, output: ``Fail."
%
%
\end{enumerate}
}
}
\end{center}
\end{figure}

By the above observations, this algorithm is correct:

\begin{fact}[Correctness of Algorithm 1]\label{fact:algorithm_correct}
If Algorithm 1 outputs {\em ``$\U$ trivially intersects $\X$"} then $\U$ indeed trivially intersects $\X$. If Algorithm 1 outputs {\em ``The only elements of $\U \cap \X$ are $\{v_1,\dots, v_s\}$ (up to scale)"} then the only elements of $\U \cap \X$ are indeed $\{v_1,\dots, v_s\}$ (up to scale).
\end{fact}
\begin{proof}
The first sentence follows directly from Observation~\ref{obs:cert}. For the second sentence, if this output occurs then $T((\field^s)^*)=S^d(\U)\cap I_d^{\perp}=\spn\{v_1^{\otimes d},\dots, v_s^{\otimes d}\}$, and $\{v_1,\dots, v_s\}$ is linearly independent. This implies by Observation~\ref{obs:find} that $v_1,\dots, v_s$ are the only elements of $\U \cap \X$ (up to scale).
\end{proof}



\section{Proof of the genericity guarantee}\label{sec:simplified}
In this section, we prove the following theorem, which immediately implies our main result Theorem~\ref{thm:intro:intersection}, a genericity guarantee for Algorithm~1. For a positive integer $d$, an $\field$-vector space $\V$ and vectors $v_1,\dots, v_d \in \V$, we define $v_1 \vee \dots \vee v_d=P_{\V,d}(v_1\otimes \dots \otimes v_d)\subseteq S^d(\V)$.
\begin{theorem}\label{thm:generic_cert:simplified}
Let $\V = \field^n$, let $d, R$ be positive integers, and let $s \in \{0,1,\dots, R\}$ be an integer. Let $\K \subseteq S^d(\V)$ be a linear subspace, and let $\X_1,\dots,\X_R \subseteq \V$ be conic varieties for which $\X_1,\dots, \X_s$ are non-degenerate of order $d-1$ and $\X_{s+1},\dots, \X_R$ are non-degenerate of order $d$. If
\ba\label{eq:ranklbsimplified}
\codim(\K) \geq R(d-1)! \binom{n+d-2}{d-1},
\ea
then for a generic choice of $v_1 \in \X_1,\dots,v_R\in \X_R$, it holds that
\ba\label{eq:setofvectors2:simplified}
\spn\{v_{a_1} \vee \dots \vee v_{a_d} : a \in [R]^{\vee d} \setminus \Delta_s\} \cap \K = \{0\},
\ea
where $\Delta_s=\{(1,\dots, 1),\dots, (s,\dots, s)\}$.
\end{theorem}
Before proving this theorem, we observe that it implies Theorem~\ref{thm:intro:intersection}.

\begin{cor}[Theorem~\ref{thm:intro:intersection}]\label{cor:generic_simplified}
Let $\V$ be an $\field$-vector space of dimension $n$, let $s \in \{0,1,\dots, R\}$ be an integer, and let $\X \subseteq \V$ be a conic variety cut out by $p= \delta \binom{n+d-1}{d}$ linearly independent homogeneous degree-$d$ polynomials $f_1,\dots, f_p \in S^d(\V^*)$ for a constant $\delta \in (0,1)$. Suppose furthermore that $\X$ is non-degenerate of order $d-1$.
Then for a linear subspace $\U \subseteq \V$ of dimension
\begin{equation}\label{eq:R}
R \le \frac{p}{ (d-1)! \cdot \binom{n+d-2}{d-1}} = \frac{\delta}{d! } \cdot (n+d-1) 
\end{equation}
spanned by a generically chosen element $(v_1,\dots, v_s,v_{s+1},\dots, v_R) \in \X^{\times s} \times \V^{\times(R-s)}$, $\U$ has only $s$ elements in its intersection with $\X$ (up to scalar multiples), and, on input any basis of $\U$, Algorithm~1 correctly outputs ``$\U$ trivially intersects $\X$" if $s=0$ and ``The only elements of $\U \cap \X$ are $\{v_1,\dots, v_s\}$ (up to scale)" if $s >0$. When $s=0$, the statement holds even without the non-degeneracy assumption on $\X$.
\end{cor}
\begin{proof}
Let $I_d=\spn\{f_1,\dots, f_p\} \subseteq S^d(\V^*)$. Since $\{f_1,\dots, f_p\}$ is linearly independent, we have that $\dim(I_d)=\codim(I_d^{\perp})=p$. By Theorem~\ref{thm:generic_cert:simplified}, for a generically chosen tuple of vectors $(v_1,\dots, v_{R}) \in \X^{\times s} \times \V^{\times (R-s)}$, it holds that $v_1,\dots, v_R$ are linearly independent and
\ba
\spn\{v_{a_1} \vee \dots \vee v_{a_d} : a \in [R]^{\vee d} \setminus \Delta_s\} \cap I_d^{\perp} = \{0\}.
\ea
In particular, $\{v_1^{\otimes d},\dots, v_R^{\otimes d}\}$ forms a basis for $S^d(\U)\cap I_d^{\perp}$. By Observation~\ref{obs:find}, $v_1,\dots, v_s$ are the only elements of $\U \cap \X$ up to scale. In Algorithm~1, the simultaneous diagonalization algorithm is applied to $T=\sum_{i=1}^s e_i \otimes P_i$, where $\{P_1,\dots, P_s\}$ is any basis of $S^d(\U)\cap I_d^{\perp}$. Let $f_1,\dots, f_R \in \V^*$ be dual to $v_1,\dots, v_R$, i.e. $f_i(v_j)=\delta_{i,j}$. Then $T=\sum_{i=1}^s z_i \otimes v_i^{\otimes d}$, where $z_i=T(f_i^{\otimes d})\in \field^s$. By Fact~\ref{fact:jennrich} this is the unique tensor rank decomposition of $T$ and it is recovered by the simultaneous diagonalization algorithm. It remains only to prove the last sentence of the corollary, which follows immediately from Theorem~\ref{thm:generic_cert:simplified} with $\K=I_d^{\perp}$ and $\X_1=\dots=\X_R=\V$. This completes the proof.
\end{proof}

Now we prove Theorem~\ref{thm:generic_cert:simplified}, for which we require the following lemma. For a positive integer $d$, an $\field$-vector space $\V$ and vectors $v_1,\dots, v_d \in \V$, we define $v_1 \vee \dots \vee v_d=P_{\V,d}(v_1\otimes \dots \otimes v_d)$. For a symmetric tensor $u \in S^d(\V^*)$, an integer $\ell \in [d-1]$, and a vector $v \in \V$, we define $v^{\otimes \ell} \hook u \in S^{d-\ell}(\V^*)$ to be the contraction of $u$ with $v^{\otimes \ell}$ in any $\ell$ of the $d$ factors. (The output will be the same regardless of which $\ell$ factors are chosen. We will pick the first $\ell$ factors for concreteness.)

\begin{lemma}\label{lemma:generic_hook:new}
Let $n\in \natural$ be a positive integer, let $d \geq 2$ be an integer, let $\ell \in [d-1]$, let $\X \subseteq \V = \field^n$ be an irreducible variety that is non-degenerate of order $d-1$, and let $\U \subseteq S^d(\V^*)$ be a linear subspace. Then for a generically chosen vector $v \in \X$, it holds  that $v^{\otimes \ell} \hook \U \subseteq S^{(d-\ell)}(\V^*)$, and
\ba\label{eq:generic_hook}
\dim(v^{\otimes \ell}\hook \U)\geq \frac{1}{\binom{n+\ell-1}{\ell}} \cdot \dim(\U).
\ea
\end{lemma}

\begin{proof}
The fact that $v^{\otimes \ell} \hook \U \subseteq S^{(d-\ell)}(\V^*)$ is obvious, so it suffices to prove the dimension bound. Since the set of $v \in \X$ that satisfy~\eqref{eq:generic_hook} is clearly Zariski open, it suffices prove that it is non-empty, i.e. that there exists a single $v \in \X$ that satisfies~\eqref{eq:generic_hook}. Since $\X$ is non-degenerate of order $d-1$, there exists $v_1,\dots, v_{m} \in \X$, where $m=\binom{n+\ell-1}{\ell}$, for which $\{v_i^{\otimes \ell} : i \in [m]\}$ forms a basis of $S^{\ell}(\V)$. Let $\{u_1,\dots,u_m\} \subseteq S^{\ell}(\V^*)$ be such that $v_i^{\otimes \ell}(u_j)=\delta_{i,j}$. Since $\U \subseteq S^d(\V^*) \subseteq S^{\ell}(\V^*) \otimes S^{(d-\ell)}(\V^*)$, any element $u \in \U$ can be written as $u=\sum_{i=1}^m u_i \otimes w_i$ for some $w_i \in S^{d-\ell}(\V^*)$. Furthermore, by construction it holds that $w_i = v_i^{\otimes \ell} \hook u \in v_i^{\otimes \ell} \hook \U$. It follows that
\ba
\U \subseteq \sum_{i=1}^m \spn\{u_i\} \otimes (v_i^{\otimes \ell} \hook \U).
\ea
Thus,
\ba
\dim(\U) \leq \sum_{i=1}^m \dim(v_i^{\otimes \ell} \hook \U),
\ea
so there exists some $i \in [m]$ for which
\ba
\dim(v_i^{\otimes \ell} \hook \U) \geq \frac{1}{\binom{n+\ell-1}{\ell}} \cdot \dim(\U).
\ea
This completes the proof.
\end{proof}

First note that it suffices to prove Theorem~\ref{thm:generic_cert:simplified} over $\complex$. Indeed, if $\field=\real$ then we can consider $\real^n$ as a subset of $\complex^n$ and let $\T_1,\dots, \T_R$ be the Zariski closures of $\X_1,\dots, \X_R$ in $\complex^n$. It is clear that $\T_i \cap \real^n=\X_i$ for each $i \in [R]$. If $\X_i\subseteq \real^n$ is non-degenerate of order $\tilde{d}$, then $\T_i\subseteq \complex^n$ is non-degenerate of order $\tilde{d}$. We can similarly replace $\K$ with $\K \otimes_{\real} \complex$ (it's dimension will not change). Furthermore, $\T:=\T_1 \times \dots \times \T_R$ is the Zariski closure of $\X:=\X_1 \times \dots \times \X_R$. For any Zariski open dense subset $\A \subseteq \T$ for which~\eqref{eq:setofvectors2:simplified} holds, it follows from Fact~\ref{fact:real_vs_complex} that $\A \cap \X \subseteq \X$ is a Zariski open dense subset for which~\eqref{eq:setofvectors2:simplified} holds. We can therefore assume $\field=\complex$ without loss of generality.

To prove Theorem~\ref{thm:generic_cert:simplified}, we will first define a total ordering of all the index tuples $(a_1, \dots, a_d) \in [R]^{\vee d}$ (recall that $(a_1, \dots, a_d) \in [R]^{\vee d}$ implies $1 \leq a_1 \leq a_2 \leq \dots \leq a_d \leq R$). 

\begin{definition}\label{def:ordering}
Given two index tuples $(a_1, \dots, a_d), (b_1, \dots, b_d) \in [R]^{\vee d}$, we use the following three rules to determine if $(a_1, \dots, a_d) \prec (b_1, \dots, b_d)$:
\begin{enumerate}
    \item $|\{a_1, \dots, a_d\}|> |\{b_1, \dots, b_d\}|$ i.e., $(a_1, \dots, a_d)$ has more distinct indices than $(b_1, \dots, b_d)$
    \item when $|\{a_1, \dots, a_d\}|= |\{b_1, \dots, b_d\}|=1,$ we use the reverse lexicographic ordering $(1,\dots, 1) \succ \dots \succ (R,\dots, R)$
    \item when $|\{a_1, \dots, a_d\}|= |\{b_1, \dots, b_d\}|\neq 1$, we use the standard lexicographic ordering. For example, $(1,1,2) \prec (1,1,3) \prec \dots \prec (1,1,R) \prec (1,2,2) \dots \prec (R-1,R,R) $. 
\end{enumerate}

\end{definition}

Note that the $s$ largest tuples with respect to this ordering are $(1,\dots, 1) \succ (2,\dots, 2) \succ \dots \succ (s,\dots, s)$.
Theorem~\ref{thm:generic_cert:simplified} is immediate from the following proposition.

\begin{prop}\label{prop:simplified}
Suppose that $\field=\complex$ and the assumptions of Theorem~\ref{thm:generic_cert:simplified} hold. Then for generically chosen $v_1 \in \X_1,\dots, v_R \in \X_R$ it holds that
\ba\label{eq:prop:simplified}
v_{i_1} \vee \dots \vee v_{i_d} \notin \spn\{v_{a_1}\vee \dots \vee v_{a_d} : (i_1,\dots, i_d) \succ (a_1,\dots, a_d) \in [R]^{\vee d}\} + \K
\ea
for all $(i_1, \dots, i_d) \in [R]^{\vee d}\setminus \Delta_s$.
\end{prop}
To see why this proposition implies Theorem~\ref{thm:generic_cert:simplified}, simply take the intersection of the Zariski open dense subsets of $\X_1\times \dots \times \X_R$ satisifying~\eqref{eq:prop:simplified} for each $(i_1, \dots, i_d) \in [R]^{\vee d}$. This intersection is again Zariski open dense in $\X_1\times \dots \times \X_R$, and~\eqref{eq:setofvectors2:simplified} holds for every tuple in this intersection.
\begin{proof}[Proof of Proposition~\ref{prop:simplified}]
For each $i \in [R]$, let $\X_{i,1},\dots, \X_{i,q_{i}}$ be the irreducible components of $\X_i$. Then the irreducible components of $\X_1 \times \dots \times \X_{R}$ are $\X_{1,j_1} \times \dots \times \X_{R,j_{R}}$ as $j_1,\dots, j_{R}$ range over $[q_1],\dots, [q_{R}]$, respectively. It suffices to prove that~\eqref{eq:prop:simplified} holds on a Zariski open dense subset of each component. To ease notation, we redefine $\X_1=\X_{1,j_1}$, \dots, $\X_R=\X_{R,j_R}$, and prove that $~\eqref{eq:prop:simplified}$ holds on a Zariski open dense subset of $\X_1 \times \dots \times \X_{R}$. We prove this by induction on $d$, starting with the base case $d=1$.

In the base case $d=1$, it holds that $\codim(\K)\leq R$, and it is sufficient to verify that $\spn\{v_{s+1},\dots, v_R\}\cap \K=\{0\}$ for generically chosen $v_{s+1} \in \X_{s+1},\dots, v_R \in \X_R$. Since the set of elements of $\X_{s+1}\times \dots \times \X_R$ satisfying this property is Zariski open, it suffices to prove that it is non-empty, which follows easily from the fact that $\X_{s+1},\dots, \X_R$ are non-degenerate (i.e., $\spn(\X_s)=\dots=\spn(\X_R)=\V$).

Proceeding inductively, suppose $d>1$. Let $I=(i_1,\dots,i_d)\in [R]^{\vee d} \setminus \Delta_s$, and let
\ba
\trouble(I)=\Big\{ a \in [R]^{\vee d}~:~ \{a_1, \dots, a_d \} = \{i_1,\dots, i_d\} \text{ and } (a_1, \dots, a_d) \prec (i_1, \dots, i_d) \Big\},
\ea
where $\prec$ is defined in Definition~\ref{def:ordering}.

For each choice of vectors $v_1\in \X_1,\dots, v_R \in \X_R$, and each $a \in [R]$, let
\ba
\U_a^{(d)}=P_{\V,d}^{\vee}(\spn\{v_a\} \otimes \V^{\otimes (d-1)})\subseteq S^d(\V),
\ea
and let
\ba
\U_{-I}= \spn\bigg\{\bigcup_{a \in [R]\setminus \{i_1,\dots, i_d\}} \U_a^{(d)} \bigg\}+\K.
\ea

Then a sufficient condition for~\eqref{eq:prop:simplified} to hold is that
\ba\label{eq:prop:simplified1}
v_{i_1}\vee \dots \vee v_{i_d} \notin \spn\{v_{a_1} \vee \dots \vee v_{a_d} : a \in T(I)\} + \U_{-I}.
\ea
Indeed, the righthand side of~\eqref{eq:prop:simplified} is contained in the righthand side of~\eqref{eq:prop:simplified1}. To complete the proof, we show that~\eqref{eq:prop:simplified1} holds for generically chosen $v_1 \in \X_1,\dots, v_R \in \X_R$. Using Chevalley's theorem, it is not difficult to show that the set of elements of $\X_1\times \dots \times \X_{R}$ satisfying~\eqref{eq:prop:simplified1} is constructible~\cite[Exercise II.3.19]{hartshorne2013algebraic}. Since any constructible set contains an open dense subset of its closure, it suffices to prove that this set is Zariski dense in $\X_1\times \dots \times \X_{R}$~\cite[Lemma 2.1]{an2012rigid}.

Note that, for any choice of $(v_j \in \X_j : j \notin I)$, it holds that $\dim(\U_a^{(d)}) \leq \binom{n+d-2}{d-1}$ for all $a \notin I$ (with equality if $v_a \neq 0$), and
\ba\label{eq:Udim}
\dim(\U_{-I})&\leq \dim(\K)+(R-k) \binom{n+d-2}{d-1} \\
			&\le \binom{n+d-1}{d} - k (d-1)! \binom{n+d-2}{d-1} \quad (\text{from }\eqref{eq:ranklbsimplified}, \text{ and } (d-1)!\ge 1),
\ea
where $k:=|\{i_1,\dots,i_d\}|$.

Let $\ell$ be the largest integer in $[d]$ for which $i_{\ell}=i_1$. We first consider the case $\ell=d$. In this case, $T(I)=\emptyset$. Then $\spn\{v^{\otimes d} : v \in \X_{i_1}\}=S^d(\V)$ since $\X_{i_1}$ is non-degenerate of order $d$. For any choice of the other vectors $(v_j\in \X_j : j \neq i_1)$, it holds that $\dim(\U_{-I})<\binom{n+d-1}{d},$ so a generic choice of $v_{i_1}\in \X_{i_1}$ satisfies~\eqref{eq:prop:simplified1}. In more details, we have demonstrated the existence of a set
\ba\label{eq:zariski_densesimplified}
\bigcup_{(v_j\in \X_j : j \neq i_1)} (v_1,\dots,v_{i_1-1}) \times \A_{(v_j\in \X_j : j \neq i_1)} \times (v_{i_1+1}, \dots, v_{d})
\ea
such that~\eqref{eq:prop:simplified1} holds for every element, where $\A_{(v_j\in \X_j : j \neq i_1)} \subseteq \X_{i_1}$ is Zariski open dense for every choice of $(v_j\in \X_j : j \neq i_1)$. The set defined in~\eqref{eq:zariski_densesimplified} is Zariski dense in $\X_1 \times \dots \times \X_R$. Indeed, for any non-empty Zariski open subset $\D \subseteq \X_1 \times \dots \times \X_R$, $\D$ must intersect some $(v_1,\dots,v_{i_1-1})\times \X_{i_1} \times (v_{i_1+1}, \dots, v_{d})$ in a non-empty open subset. Since $\A_{(v_j\in \X_j : j \neq i_1)} \subseteq \X_{i_1}$ is Zariski open dense, it follows that $\D$ must intersect $(v_1,\dots,v_{i_1-1}) \times \A_{(v_j\in \X_j : j \neq i_1)} \times (v_{i_1+1}, \dots, v_{d})$, and hence it must intersect~\eqref{eq:zariski_densesimplified}. Since $\D$ was an arbitrary non-empty open subset of $\X_1\times \dots \times \X_R$, it follows that the set~\eqref{eq:zariski_densesimplified} is Zariski dense in $\X_1 \times \dots \times \X_R$. This completes the proof in the case $\ell=d$.

Henceforth we assume $1 \le \ell \le d-1$, and prove that~\eqref{eq:prop:simplified1} holds for generically chosen $v_1\in \X_1,\dots, v_R \in \X_R$ by applying the inductive hypothesis with degree $(d-\ell)$. For all $a \in T(I)$ it holds that $a_1=\dots=a_{\ell}=i_1$. Let $I_{-\ell}=(i_{\ell+1},\dots, i_d)$. Note that for any choice of $v_1\in \X_1,\dots, v_R \in \X_R$, it holds that
\ba
\spn\{v_{a_1} \vee  \dots& \vee v_{a_d} : a \in T(I)\} \subseteq \spn\{v_{i_1}^{\vee \ell} \vee v_{b_1} \vee \dots \vee v_{b_{d-\ell}} : b \in T(I_{-\ell})\}+\W\\
&\subseteq \spn\{v_{i_1}^{\vee \ell} \vee v_{b_1} \vee \dots \vee v_{b_{d-\ell}} : I_{-\ell} \succ b \in \{i_{\ell+1},\dots, i_d\}^{\vee (d-\ell)}\}+\W,\label{eq:prop:simplified2}
\ea
where
\ba
\W=P_{\V,d}^{\vee}(\spn\{v_{i_1}^{\otimes (\ell+1)}\} \otimes \V^{\otimes(d-\ell-1)}).
\ea
Indeed, the first line follows from the fact that for any $a \in T(I)$, if $a_{\ell+1}\neq i_1$ then $(a_{\ell+1},\dots, a_d) \in T(I_{-\ell}),$ and if $a_{\ell+1}=i_1$ then $a_1=\dots=a_{\ell+1}=i_1$ and $v_{i_1}^{\vee (\ell+1)} \vee v_{a_{\ell+2}} \vee \dots \vee v_{a_d} \in \W$. The second line is immediate. (Here, the ordering on $\{i_{\ell+1},\dots, i_d\}^{\vee (d-\ell)}$ is that of Definition~\ref{def:ordering} under the bijection between $\{i_{\ell+1},\dots, i_d\}$ and $[k-1]$ that sends the $j$-th distinct element of $(i_{\ell+1},\dots, i_d)$ to $j$. Recall that $k=\abs{\{i_1,\dots, i_d\}}$, so $k-1=\abs{\{i_{\ell+1},\dots,i_d\}}$.)

By~\eqref{eq:prop:simplified2}, to verify~\eqref{eq:prop:simplified1} it suffices to prove that
\ba\label{eq:prop:simplified2new}
v_{i_1}^{\vee \ell} \vee v_{i_{\ell+1}}\vee \dots \vee v_{i_d} \notin \spn\{v_{i_1}^{\vee \ell} \vee v_{b_1}\vee \dots \vee v_{b_{d-\ell}} : I_{-\ell} \succ b\} + \W + \U_{-I},
\ea
or equivalently,
\ba\label{eq:prop:simplified3}
v_{i_{\ell+1}}\vee \dots \vee v_{i_d} \notin \spn\{v_{b_1}\vee \dots \vee v_{b_{d-\ell}} : I_{-\ell} \succ b\} + \U_{i_1}^{(d-{\ell})} + (v_{i_1} \hook \U_{-I}^{\perp})^{\perp},
\ea
where $(\cdot)^\perp$ denotes the orthogonal complement in the dual symmetric space (for example, $\U_{-I}^\perp \in S^d(\V^*)$), and
\ba
\U_{i_1}^{(d-{\ell})}:=P_{\V,(d-\ell)}^{\vee}(\spn\{v_{i_1}\} \otimes \V^{\otimes (d-\ell-1)})\subseteq S^{(d-\ell)}(\V).
\ea
To see this equivalence, first note that~\eqref{eq:prop:simplified2new} is equivalent to
\ba\label{eq:prop:simplified4}
v_{i_1}^{\otimes \ell} \otimes (v_{i_{\ell+1}}\vee \dots \vee v_{i_d}) \notin \spn\{v_{i_1}^{\otimes \ell} \otimes (v_{b_1}\vee \dots \vee v_{b_{d-\ell}}) &: I_{-\ell} \succ b\} \\
&+\spn\{v_{i_1}^{\otimes \ell}\}\otimes \U_{i_1}^{(d-\ell)} + \U_{-I}+ \Z,
\ea
where $\Z=\ker(P_{\V,d}^{\vee}) \subseteq \V^{\otimes d}$. (This follows from the basic fact that for any vector space $\V$, linear subspace $\U \subseteq \V$, projection $P : \V \rightarrow \V$, and vector $v \in \V$, it holds that  $P(v) \in P(\U)$ if and only if $v \in P(\U)+ \ker(P)$. For us, the righthand side of~\eqref{eq:prop:simplified2new} plays the role of $\U$, $\I_{\V}^{\otimes \ell}\otimes P_{\V,(d-\ell)}^{\vee}$ plays the role of $P$, and $\V^{\otimes d}$ plays the role of $\V$.) So we just need to prove that~\eqref{eq:prop:simplified4} and~\eqref{eq:prop:simplified3} are equivalent. This follows from another basic fact: For vector spaces $\V_1, \V_2$ over the same field, a vector $v_2 \in \V_2$, and a linear subspace $\U \subseteq \V_1 \otimes \V_2$, the set of vectors $v_1 \in \V_1$ for which $v_1 \otimes v_2 \in \U$ is precisely the linear subspace $(v_2 \hook \U^{\perp})^{\perp}\subseteq \V_1$. Indeed,
\ba
v_1 \otimes v_2 \in \U \iff ({v_1 \otimes v_2})(\U^{\perp})=0 \iff {v_1}(v_2 \hook \U^{\perp})=0\iff v_1 \in (v_2 \hook \U^{\perp})^{\perp}.
\ea
For us, $\V^{\otimes \ell}$ plays the role of $\V_1$, $S^{(d-\ell)}(\V)$ plays the role of $\V_2$, and $\U_{-I}+\Z$ plays the role of $\U$.

At this point we have shown that to prove the proposition it suffices to verify that~\eqref{eq:prop:simplified3} holds for generically chosen $v_1 \in \X_1,\dots, v_R \in \X_R$. To apply the induction hypothesis, we just need to verify that the codimension of $ \U_{i_1}^{(d-{\ell})} + (v_{i_1} \hook \U_{-I}^{\perp})^{\perp}$ is generically large enough. By Lemma~\ref{lemma:generic_hook:new} and the upper bound on $\dim(\U_{-I})$ derived in~\eqref{eq:Udim}, for a generically chosen vector $v_{i_1}\in \X_{i_1}$ it holds that
\ba
\dim(v_{i_1} \hook \U_{-I}^{\perp}) \geq \frac{k(d-1)! \binom{n+d-2}{d-1}}{\binom{n+\ell-1}{\ell}}.
\ea
Let $\B_{(v_j \in \X_j : j \notin I)} \subseteq \X_{i_1}$ be the Zariski open dense subset on which this holds (since $\U_{-I}$ depends on $(v_j \in \X_j : j \notin I)$, $\B_{(v_j \in \X_j : j \notin I)}$ also does, and we keep track of this with the subscript). Then for any $v_{i_1} \in \B_{(v_j \in \X_j : j \notin I)}$ it holds that
\ba
\codim(\U_{i_1}^{(d-{\ell})} + (v_{i_1} \hook \U_{-I}^{\perp})^{\perp})& \geq \frac{k(d-1)! \binom{n+d-2}{d-1}}{\binom{n+\ell-1}{\ell}}-\binom{n+d-\ell-2}{d-\ell-1}\\
&\ge \frac{k (d-1)!}{\binom{d-1}{\ell}} \cdot\binom{n+d-\ell-2}{d-\ell-1} -\binom{n+d-\ell-2}{d-\ell-1}\\
&\ge (k-1) \cdot (d-\ell-1)! \cdot \binom{n+d-\ell-2}{d-\ell-1},
\ea
where the second line follows from $\binom{n+d-2}{d-1} \cdot \binom{d-1}{\ell}\ge \binom{n+d-\ell-2}{d-\ell-1} \cdot \binom{n+\ell-1}{\ell}$. By the induction hypothesis, there is a Zariski open dense subset $\C_{(v_j \in \X_j : j \notin I_{-\ell})}\subseteq \times_{j \in I_{-\ell}} \X_j$ for which~\eqref{eq:prop:simplified3} holds.

At this point we have proven that~\eqref{eq:prop:simplified3} holds for every element of the set
\ba~\label{eq:zariski21simplified}
\bigcup_{\substack{(v_j\in \X_j : j \notin I)\\ v_{i_1} \in \B_{(v_j : j \notin I)}}} (v_j \in \X_j : j \notin I) \times (v_{i_1}) \times \C_{(v_j : j \notin I_{-\ell})}.
\ea
It remains only to check that this set is Zariski dense in $\X_1 \times \dots \times \X_R$. To this end, let $\D\subseteq \X_1 \times \dots \times \X_R$ be a non-empty open subset. Then there exists $(v_j \in \X_j : j \notin I)$ for which
\ba
\D_1:= \D \cap \bigg((v_j \in \X_j : j \notin I) \times \big(\times_{j \in I} \X_j\big)\bigg)
\ea
is Zariski open dense in $(v_j \in \X_j : j \notin I)\times \left(\times_{j \in I} \X_j\right)$. For this choice of $(v_j \in \X_j : j \notin I)$, there exists $v_{i_1} \in \B_{(v_j \in \X_j : j \notin I)}$ for which 
\ba
\D_2 := \D_1 \cap \bigg((v_j \in \X_j : j \notin I_{-\ell}) \times \big(\times_{j \in I_{-\ell}} \X_j\big)\bigg)
\ea
is Zariski open dense in $(v_j \in \X_j : j \notin I_{-\ell})\times \left(\times_{j \in I_{-\ell}} \X_j\right)$. Thus,
\ba
\D_2 \cap \left((v_j \in \X_j : j \notin I_{-\ell}) \times \C_{(v_j \in \X_j : j \notin I_{-\ell})}\right) \neq \emptyset,
\ea
since $\C_{(v_j \in \X_j : j \notin I_{-\ell})}\subseteq \times_{j \in I_{-\ell}} \X_j$ is Zariski open dense.
So $\D$ intersects the set defined in~\eqref{eq:zariski21simplified} non-trivially. Since $\D$ was an arbitrary non-empty open subset of $\X_1\times \dots \times \X_R$, it follows that the set~\eqref{eq:zariski21simplified} is Zariski dense in $\X_1 \times \dots \times \X_R$. This completes the proof.
\end{proof}



\section{Application to determining entanglement of a linear subspace} \label{sec:entanglement}

In the context of quantum information theory, there are many scenarios in which it is useful to determine whether or not a linear subspace $\U$ intersects a conic variety $\X$. For example, when ${\field=\complex}$, for positive integers $n_1$ and $n_2$ and a positive integer $r < \min\{n_1,n_2\}$, determining whether or not a linear subspace $\U \subseteq \field^{n_1}\otimes \field^{n_2}$ intersects the determinantal variety $\X_r$ has found applications in quantum entanglement theory (e.g., the problems of constructing entanglement witnesses and determining whether or not a mixed quantum state is separable) and quantum error correction, among many others \cite{Hor97,BDMSST99,ATL11,CS14,HM10} (see Section~\ref{eq:XR} for the definition of $\X_r$).

If $\U$ trivially intersects $\X_r$, then we say that $\U$ is \textit{$r$-entangled} (or just \emph{entangled} if $r = 1$). Other relevant examples include \textit{completely entangled subspaces}, subspaces of $\field^{n_1}\otimes \dots \otimes \field^{n_m}$ which trivially intersect the set of separable tensors $\X_{\setft{Sep}}$; and~\textit{genuinely entangled subspaces}, subspaces of $\field^{n_1}\otimes \dots \otimes \field^{n_m}$ which trivially intersect the set of biseparable tensors $\X_{B}$.
We will also consider subspaces avoiding the set of tensors of slice rank $1$, $\X_{S}$. While it is not clear if this last example has quantum applications, we include it because $\X_S$ has found several recent applications in theoretical computer science~\cite{petrov2016combinatorial,kleinberg2016growth,blasiak2016cap,naslund2017upper,fox2017tight}.

In particular, determining whether $\U\subseteq \field^{n_1}\otimes \field^{n_2}$ is $r$-entangled is NP-hard~\cite{buss1999computational}. A slightly easier problem is: given the promise that either $\U\cap \X_r$ contains a non-zero element, or else $\U$ is $\epsilon$\textit{-far} from $\X_r$ in the sense that
\ba
\norm{v-u}> \epsilon \norm{v}
\ea
for all $u \in \U$ and $v \in \X_r$, determining which of these two possibilities is the case. Here, ${\norm{\cdot}=\ip{\cdot}{\cdot}^{1/2}}$ is the 2-norm. There is strong evidence that solving this problem should also take super-polynomial time in $\min\{n_1,n_2\}$ in the worst case~\cite[Corollary 14]{HM10}. To our knowledge, the best known algorithm for solving this problem takes $\exp(\tilde{O}({\sqrt{n_1}/\epsilon}))$ time in the worst case when $r=1$ and $n_1=n_2$~\cite{barak2017quantum}.

Despite these hardness results, our algorithm runs in polynomial time, and determine whether a subspace $\U\subseteq \field^{n_1}\otimes \field^{n_2}$, of dimension up to a constant multiple of the maximum possible, is $r$-entangled. We obtain analogous results for completely and genuinely entangled subspaces.

In these settings (and in contrast to the decomposition setting described in the next section), we are not concerned with uniqueness, i.e. determining whether a found element $v \in \U\cap \X$ (or collection of elements) is the \textit{only} element of $\U\cap \X$. For reducible varieties, we can use this flexibility to our advantage, and employ a variant of Algorithm~1 which has better scaling. In short, this adaptation simply runs Algorithm~1 on each irreducible component of $\X$. If $\X_1,\dots, \X_k$ are cut out by homogeneous polynomials $p_1,\dots, p_k$ of degrees $d_1,\ldots, d_k$, then $\X=\X_1 \cup \dots \cup \X_k$ is naively cut out by the homogeneous polynomial $p_1\cdots p_k$ of degree $d_1\cdots d_k$. The main advantage of our adapted algorithm in this setting is that it avoids this blow-up in the degree. We call this adapted algorithm \textit{Algorithm 2}, and describe it formally below.

\begin{figure}[htbp]
\begin{center}
\fbox{\parbox{0.98\textwidth}{

\begin{center}\textbf{\Large{Algorithm 2: Determining whether $\U\cap \X=\{0\}$.}}\end{center}
\textbf{Input:} A basis $\{u_1,\dots, u_{R}\}$ for a linear subspace $\U \subseteq \V=\field^n$, and for each $i \in [k]$ a collection of homogeneous degree-$d_i$ polynomials $f_{i,1},\dots, f_{i,p_i}$ that cut out the $i$-th irreducible component of a conic variety $\X=\X_1\cup\dots \cup \X_k \subseteq \V$.
\begin{enumerate}
\item For each $i \in [k]$, run Algorithm 1 on input $\{u_1,\dots,u_R\}$ and polynomials $f_{1,i},\dots, f_{p_i,i}$ cutting out $\X_i$, and output any non-zero elements of $\U \cap \X_i$ found by Algorithm~1.
\item If all of these output ``$\U$ trivially intersects $\X_i$," then output ``$\U$ trivially intersects $\X$.''
\item Otherwise, output ``Fail."
%
%
\end{enumerate}
}
}
\end{center}
\end{figure}

Corollary~\ref{cor:generic_simplified} implies the following genericity guarantee for Algorithm 2:

\begin{theorem}\label{thm:general_cert}
Let $n, d_1,\dots, d_k$ be positive integers, let $\delta_1,\dots, \delta_k \in (0,1)$, let $\V$ be an $\field$-vector space of dimension $n$, and let $\X \subseteq \V$ be a conic variety with irreducible components $\X_1,\dots, \X_k$, such that each $\X_i$ is non-degenerate of order $d_i-1$ and is generated by $p_i= \delta \binom{n+d_i-1}{d_i}$ linearly independent homogeneous degree-$d_i$ polynomials. If $\U\subseteq \V$ is a generically chosen linear subspace, possibly containing a generically chosen ``planted" element of $\X$, of dimension
\ba\label{eq:general_cert}
R:=\dim(\U)\leq \min_{i \in [k]} \Big( \frac{\delta_i}{d_i! } \cdot (n+d_i -1) \Big)  
\ea
then Algorithm 2 either certifies that $\U \cap \X =\{0\}$, or else produces the planted element of $\U \cap \X$.
\end{theorem}
In more details, Theorem~\ref{thm:general_cert} asserts that the following two statements hold:
\begin{enumerate}
\item For every positive integer $R$ satisfying~\eqref{eq:general_cert}, there exists a Zariski open dense subset ${\A\subseteq \V^{\times R}}$ such that for all $(v_1,\dots, v_{R}) \in \A$, the linear subspace $\U:=\spn\{v_1,\dots, v_{R}\}$ trivially intersects $\X$, and Algorithm 2 correctly outputs ``$\U$ trivially intersects $\X$."
\item For every positive integer $R$ satisfying~\eqref{eq:general_cert}, there exists a Zariski open dense subset ${\B\subseteq \X \times \V^{\times {R-1}}}$ such that for all $(v_1,\dots, v_{R}) \in \B$, Algorithm~2 outputs $v_1\in \U \cap \X$.
\end{enumerate}

\begin{proof}[Proof of Theorem~\ref{thm:general_cert}]
By Corollary~\ref{cor:generic_simplified}, for each $i \in [k]$ there exists a Zariski open dense subset $\A_i \subseteq \V^{\times R}$ such that for all $(v_1,\dots, v_{R}) \in \A_i$, the linear subspace $\U:=\spn\{v_1,\dots, v_{R}\}$ trivially intersects $\X_i$, and Algorithm 1 correctly outputs ``$\U$ trivially intersects $\X_i$." We can therefore take $\A:=\A_1 \cap \dots \cap \A_k$ to obtain the first statement above. By Corollary~\ref{cor:generic_simplified}, for each $i \in [k]$ there exists a Zariski open dense subset $\B_i \subseteq \X_i \times \V^{\times {R-1}}$ such that for all $(v_1,\dots, v_{R}) \in \B_i$, Algorithm 1 correctly outputs ``$v_1$ is the only element of $\U \cap \X_i$." The theorem follows by taking $\B=\B_1 \cup \dots \cup \B_k$, which is an open dense subset of $\X \times \V^{\times R-1}$.
\end{proof}

\begin{cor}\label{cor:XR}
Let $n_1, n_2$ be positive integers, let $r < \min \{n_1,n_2\}$ be a positive integer, and let ${\V=\field^{n_1}\otimes\field^{n_2}}$. If $\U\subseteq \V$ is a generically chosen linear subspace, possibly containing a generically chosen ``planted" element of $\X_r$, of dimension
\ba\label{eq:dim_upper}
\dim(\U)\leq \frac{\binom{n_1}{r+1}\binom{n_2}{r+1}}{(r+1)! \binom{n_1 n_2+r}{r+1}} (n_1 n_2 + r),
\ea
then (in time $(n_1 n_2)^{O(r)}$) Algorithm 2 either certifies that $\U\cap \X_r=\{0\}$ or else produces the planted element of $\U \cap \X_r$. Note that the righthand side of~\eqref{eq:dim_upper} is $\Omega_r(n_1n_2)$ for any fixed $r$.
\end{cor}

Trivially, $\dim(\U)\leq n_1 n_2$ for any $r$-entangled subspace, so the upper bound~\eqref{eq:dim_upper} is a quite mild condition on $\dim(\U)$.\footnote{Over $\complex$, it is a standard fact that the maximum dimension of an $r$-entangled subspace is $(n_1-r)(n_2-r)$~\cite{harris2013algebraic,CMW08}. Over $\real$, there can be larger $r$-entangled subspaces. For example, the $2$-dimensional subspace
\ba
\spn\{e_1 \otimes e_1, e_1 \otimes (e_1+e_2)-e_2\otimes (2e_1 + e_2)\} \subseteq \real^2 \otimes \real^2
\ea
is $1$-entangled. The maximum dimension of a real $r$-entangled subspace does not seem to be known in general. See e.g.~\cite{petrovic1996spaces,rees1996linear} for work in this direction.}

\begin{proof}[Proof of Corollary~\ref{cor:XR}]
Recall from Section~\ref{sec:XR} that $\X_r$ is a conic variety cut out by $p=\binom{n_1}{r+1}\binom{n_2}{r+1}$ homogeneous polynomials of degree $d=r+1$, and it has no equations in degree $r$ (see Section~\ref{sec:XR}). Thus, the statement follows from Theorem~\ref{thm:general_cert}.
\end{proof}

We can obtain similar corollaries for the varieties $\X_{\setft{Sep}}, \X_B$ and $\X_S$, introduced in Section~\ref{sec:XR}, as follows. We omit the proofs, as they are very similar to the proof of Corollary~\ref{cor:XR}.


\begin{cor}
Let $m$ be a positive integer, let $n_1,\dots, n_m$ be positive integers, and let $\V=\field^{n_1}\otimes\dots \otimes \field^{n_m}$. If $\U\subseteq \V$ is a generically chosen linear subspace, possibly containing a generic ``planted" element of $\X_{\setft{Sep}}$, of dimension
\ba\label{eq:dim_upper_completely}
\dim(\U)\leq \frac{\binom{n_1\cdots n_m+1}{2}-\left[\binom{n_1+1}{2}\cdots\binom{n_m+1}{2}\right]}{ n_1\cdots n_m} = \frac{1}{2}  (n_1 \dots n_m +1) - \frac{1}{2^m} (n_1+1) \dots (n_m+1),
\ea
then (in time $O(n_1\cdots n_m)$) Algorithm 2 either certifies that $\U \cap \X_{\setft{Sep}}=\{0\}$, or else produces the planted element of $\U \cap \X_{\setft{Sep}}$. Note that the righthand side of~\eqref{eq:dim_upper_completely} is $\Omega(n_1\cdots n_m)$.
\end{cor}

\begin{cor}
Let $m$ be a positive integer, let $n_1,\dots, n_m$ be positive integers, and let $\V=\field^{n_1}\otimes\dots \otimes \field^{n_m}$. If $\U\subseteq \V$ is a generically chosen linear subspace, possibly containing a generically chosen ``planted" element of $\X_B$, of dimension
\ba\label{eq:dim_upper_biseparable}
\dim(\U)\leq {\left(\frac{1}{n_1\cdots n_m}\right)\min_{\substack{S \subseteq [m]\\ 1 \leq \abs{S}\leq m-1}} \binom{\prod_{i \in S} n_i}{2}\binom{\prod_{j \in [m]\setminus S} n_j}{2}},
\ea
then (in time $O(2^m n_1\cdots n_m)$) Algorithm 2 either certifies that $\U \cap \X_{B}=\{0\}$, or else produces the planted element of $\U \cap \X_{B}$.  Note that the righthand side of~\eqref{eq:dim_upper_biseparable} is $\Omega(n_1 \cdots n_m)$.
\end{cor}

\begin{cor}\label{cor:slice}
Let $m$ be a positive integer, let $n_1,\dots, n_m$ be positive integers, and let $\V=\field^{n_1}\otimes\dots \otimes \field^{n_m}$. If $\U\subseteq \V$ is a generically chosen linear subspace, possibly containing a generically chosen ``planted" element of $\X_{S}$, of dimension
\ba\label{eq:dim_upper_slice}
\dim(\U)\leq {\left(\frac{1}{n_1\cdots n_m}\right)\min_{i \in [m]} \binom{n_i}{2}\binom{\prod_{j \in [m]\setminus \{i\}} n_j}{2}},
\ea
then (in time $O(m n_1\cdots n_m)$) Algorithm 2 either certifies that $\U \cap \X_{S}=\{0\}$, or else produces the planted element of $\U \cap \X_{S}$. Note that the righthand side of~\eqref{eq:dim_upper_slice} is $\Omega(n_1 \cdots n_m)$.
\end{cor}

In all of these corollaries, the upper bound on $\dim(\U)$ is $\Omega(n_1\cdots n_m)$. Trivially, $\dim(\U)\leq n_1 \cdots n_m$ for any subspace, so this is a very mild condition on the dimension. \footnote{The maximum dimension of a completely entangled subspace over $\complex$ is
\ba
n_1\cdots n_m-\sum_{i=1}^m(n_i-1)-1.
\ea
The maximum dimension of a genuinely entangled subspace over $\complex$ is
\ba
\min_{\substack{S\subseteq [m]\\ 1\leq \abs{S} \leq \floor{m/2}}}\left(\prod_{i \in S} n_i-1\right)\left(\prod_{j \in [m]\setminus S} n_j-1\right).
\ea
The maximum dimension of a subspace that trivially intersects $\X_S$ over $\complex$ is
\ba
\min_{i \in [m]}\left(n_i-1\right)\left(\prod_{j \in [m]\setminus \{i\}} n_j-1\right).
\ea
The maximum dimension of such subspaces over $\real$ can be greater in general.}


%
%

\section{Application to low-rank decompositions over varieties}\label{sec:tensors}

\bnote{All Theorems/Corollaries in this section tweaked}
Let $\V, \W$ be arbitrary $\field$-vector spaces, and let $\X \subseteq \V$ be a non-degenerate conic variety.
In this section, we study $(\X,\W)$\textit{-decompositions},
which express a given $T \in \V \otimes \W$ in the form
\ba\label{eq:XW_decomp}
T=\sum_{i \in [R]} v_i \otimes w_i
\ea
for some $v_1,\dots, v_R \in \X$ and $w_1,\dots, w_r \in \W$, with $R$ as small as possible. 
%
We call the smallest possible $R$ for which there exists an $(\X,\W)$-decomposition of $T$ with $R$ summands the~\textit{$(\X, \W)$-rank} of $T$, and say that an $(\X,\W)$-decomposition~\eqref{eq:XW_decomp} of $T$ is the \textit{unique $(\X,\W)$-rank decomposition} of $T$ if it is an $(\X,\W)$-rank decomposition of $T$ and the only other $(\X,\W)$-rank decompositions of $T$ are those formed by permuting the $R$ summands of the decomposition. See Section~\ref{sec:variety_decomp} for further background.

In this section, we show that Algorithm~1 can be used to compute the (unique) $(\X,\W)$-rank decomposition of a generically chosen tensor $T$ of small enough $(\X,\W)$-rank. We apply these results to the case of tensor rank decompositions and $r$-aided rank decompositions. First note that Observation~\ref{obs:find} and Fact~\ref{fact:jennrich} yield a sufficient condition for a given $(\X,\W)$-decomposition of $T$ to be the unique $(\X,\W)$-rank decomposition of $T$:

\begin{prop}[Sufficient condition for uniqueness]\label{prop:suffsimplified}
Let $\X \subseteq \V:=\field^n$ be a conic variety cut out by $p$ linearly independent homogeneous degree-$d$ polynomials $f_1,\dots, f_p\in S^d(\V^*)$, let $I_d=\spn\{f_1,\dots, f_p\}\subseteq S^d(\V^*)$, let $T \in \V \otimes \W$, and let
\ba\label{eq:suff_decomp}
T=\sum_{a\in [R]} v_a \otimes w_a
\ea
be an $(\X,\W)$-decomposition of $T$. If $\{w_1,\dots, w_R\}$ is linearly independent, $\{v_1,\dots, v_R\}$ is linearly independent, and
\ba\label{eq:suff_phi}
S^d(T(\W^*)) \cap I_d^{\perp}=\spn\{v_1^{\otimes d},\dots, v_s^{\otimes d}\},
\ea
then~\eqref{eq:suff_decomp} is the unique $(\X,\W)$-rank decomposition of $T$, and furthermore this decomposition can be recovered from $T$ in $n^{O(d)}$ time using Algorithm~1.
\end{prop}
\begin{proof}
By Observation~\ref{obs:find} and Fact~\ref{fact:jennrich}, it holds that $v_1,\dots, v_R$ are the only elements of $T(\W^*) \cap \X$ (up to scale), and these are recovered in $n^{O(d)}$ time by Algorithm~1 (see the proof of Corollary~\ref{cor:generic_simplified} for more details). It remains only to recover the vectors $\{w_1,\dots, w_R\}$ up to scale. Since $\{w_1,\dots, w_R\}$ is linearly independent, the $(\X,\W)$-rank of $T$ is equal to $R$. It follows that any $(\X,\W)$-rank decomposition must involve (scalar multiples of) $v_1,\dots, v_R$. Since $v_1,\dots, v_R$ are linearly independent, they uniquely determine $w_1,\dots, w_R$. To recover $w_1,\dots, w_R$, let $f_1,\dots, f_r \in \V^*$ be dual to $v_1,\dots, v_R$, i.e. satisfy $f_i(v_j)=\delta_{i,j}$. Then $w_i=T(f_i)$ for all $i \in [R]$. This completes the proof.
\end{proof}



Combining Proposition~\ref{prop:suffsimplified} with Theorem~\ref{thm:generic_cert:simplified}, we obtain the following genericity guarantee for using Algorithm~1 to recover $(\X,\W)$ decompositions (this is slightly more general than Theorem~\ref{thm:decomp} in the Introduction).



\begin{theorem}\label{cor:generic_decomp}
Let $\X \subseteq \V=\field^n$ be an irreducible conic variety cut out by $p=\delta \binom{n+d-1}{d}$ linearly independent homogeneous degree-$d$ polynomials for constants $d \ge 2$ and $\delta\in(0,1)$. Suppose furthermore that $\X$ is non-degenerate of order $d-1$. 
Then for a tensor $T \in \V \otimes \W$ of the form
\ba\label{eq:XW}
T=\sum_{a \in [R]} v_a \otimes w_a,
\ea
where $v_1,\dots, v_R$ are chosen generically from $\X$, $\{w_1,\dots, w_R\}$ is linearly independent, and
\ba\label{eq:XW_R}
R \leq \min\Big\{ \frac{\delta}{d!}\cdot (n+d-1), \dim(\W)\Big\},
\ea
the following holds: In $n^{O(d)}$ time, Algorithm 1 can be used to recover the decomposition~\eqref{eq:XW} and certify that this is the unique $(\X,\W)$-rank decomposition of $T$.
%
\end{theorem}

In particular, this theorem proves that the $(\X,\W)$-rank decomposition of a generically chosen tensor $T \in \V \otimes \W$ of $(\X,\W)$-rank $R$ upper bounded by~\eqref{eq:XW_R} can be recovered and certified as unique by our algorithm. Note that, since $\X$ is non-degenerate, a generically chosen collection of $R$ elements of $\X$ will be linearly independent. Hence, one can alternatively set $\W=\V$ and also choose $\{w_1,\dots, w_R\}$ generically from $\X$, and the same uniqueness/recovery results hold. More generally, one can choose $\{w_1,\dots, w_R\}$ generically from any non-degenerate variety $\Y \subseteq \W$, and the same uniqueness/recovery results hold.

By letting $\X=\X_1=\{u \otimes v : u \in \field^{n_1}, v \in \field^{n_2}\}$, we obtain the corollary for recovering unique decompositions of order-$3$ tensors with potentially unequal dimensions.
\begin{cor}\label{cor:tensor_decomp}
Let $n_1,n_2,n_3$ be positive integers. For a generically chosen tensor $T \in \field^{n_1}\otimes \field^{n_2} \otimes \field^{n_3}$ of tensor rank
\ba
R \leq \min\Bigg\{\frac{1}{4}(n_1-1)(n_2-1), n_3\Bigg\},
\ea
in $(n_1 n_2)^{O(1)}$ time Algorithm 1 can be used to recover the tensor rank decomposition of $T$ and certify that it is unique.
\end{cor}
The above corollary shows that when $n_3 =\Omega(n_1 n_2)$, we can go all the way up to rank $\Omega(n_1 n_2),$ which is the maximum possible rank up to constants.

Letting $\X=\X_1=\{u \otimes v : u,v \in \field^n\}$ and $\W=\field^{n}\otimes \field^{n}$, and choosing $w_1,\dots, w_R$ generically from $\X_1$ (as in the discussion following Theorem~\ref{cor:generic_decomp}), we obtain the following corollary for order-$4$ tensors (this is a special case of Corollary~\ref{cor:general_tensor_decomp} below).

\begin{cor}\label{cor:tensor_decompB}
For any positive integer $n$, and a generically chosen tensor $T \in \field^{n}\otimes \field^{n} \otimes \field^{n} \otimes \field^{n}$ of tensor rank
\ba
R \leq  \frac{(n-1)^2}{4},
\ea
in $n^{O(1)}$ time Algorithm 1 can be used to recover the tensor rank decomposition of $T$ and certify that it is unique.
\end{cor}
%

More generally, we have the following corollary for tensors of arbitrary order:

\begin{cor}\label{cor:general_tensor_decomp}
Let $n$ be a positive integer, and let $m \geq 3$ be an integer. Then for a generically chosen tensor $T\in (\field^{n})^{\otimes m}$ of tensor rank
\ba\label{eq:R_bound}
R \leq \min\left\{ n^{\floor{m/2}}, \frac{n^{\ceil{m/2}}+1}{2}-\frac{(n+1)^{\ceil{m/2}}}{2^{\ceil{m/2}}}\right\},
\ea
in $n^{O(m)}$ time Algorithm 1 can be used to recover the tensor rank decomposition of $T$ and certify that it is unique.
\end{cor}
It can be shown that for all $n \ge 8$, the bound~\eqref{eq:R_bound} translates to 
\ba
R \leq \begin{cases} n^{\floor{m/2}} & \text{ if } m \text{ is odd} \\ \frac{n^{\ceil{m/2}}+1}{2}-\frac{(n+1)^{\ceil{m/2}}}{2^{\ceil{m/2}}} & \text{ if } m \text{ is even},\end{cases}
\ea

which is $\Omega(n^{\lfloor m/2 \rfloor})$ as $n$ grows. 
For even order $m$, 
our results extend to non-symmetric tensors the bounds known for symmetric decompositions~\cite{MSS, BCPV}
(see also \cite{Vij20} for related references). In particular, we are not aware of any existing genericity guarantees (prior to our work) for non-symmetric tensors of even order $m$ that work for rank $R = \Omega(n^{{m/2}})$.\footnote{For odd $m$, a variant of Harshman's algorithm~\cite{Har72} works for rank  $O(n^{(m-1)/2})$ (see e.g., \cite{BCMV}).}


\begin{proof}[Proof of Corollary~\ref{cor:general_tensor_decomp}]
We prove the statement by regarding tensor decompositions in $(\field^n)^{\otimes m}$ as $(\X,\W)$-decompositions, where $\X=\X_{\setft{Sep}}\subseteq (\field^n)^{\otimes \ceil{m/2}}$ and the elements of $\W$ appearing in the decomposition are constrained to be in $\tilde{\X}_{\setft{Sep}}\subseteq (\field^n)^{\otimes \floor{m/2}}$.

Since $\X_{\setft{Sep}} \subset (\field^{n})^{\otimes \ceil{m/2}}$ is non-degenerate and cut out by
\ba
\binom{n^{\ceil{m/2}}+1}{2}-\binom{n+1}{2}^{\ceil{m/2}}
\ea
many linearly independent homogeneous polynomials of degree 2 in $n^{\ceil{m/2}}$ variables, it follows from Theorem~\ref{cor:generic_decomp} (and the subsequent discussion) that our algorithm recovers unique tensor decompositions of rank
\ba
R \leq \min\left\{ n^{\floor{m/2}}, \frac{n^{\ceil{m/2}}+1}{2}-\frac{(n+1)^{\ceil{m/2}}}{2^{\ceil{m/2}}}\right\}.
\ea
This completes the proof.
%
%
%
%
%
\end{proof}

We obtain an analogous result for symmetric tensor decompositions. For a symmetric tensor $T\in S^m(\field^n)$, a \textit{symmetric decomposition} of $T$ is a decomposition of the form $T=\sum_{a\in [R]} \alpha_a v_a^{\otimes m}$ for some $\alpha_1,\dots, \alpha_R \in \field$ and $v_1,\dots, v_R \in \field^n$ (in the terminology introduced in Section~\ref{sec:variety_decomp}, these exactly correspond to $\X^{\vee}_{\setft{Sep}}$-decompositions). The \textit{Waring rank} of $T$ is the minimum number of terms needed in the decomposition, and a Waring rank decomposition of $T$ is said to be unique if the only other Waring rank decompositions of $T$ are those obtained by permuting terms in the sum. We say that a property holds for a \textit{generically chosen} symmetric tensor $T \in S^m(\field^n)$ of Waring rank at most $R$ if the property holds for every tensor of the form $T=\sum_{a\in [R]} \alpha_a v_a^{\otimes m}$, where $\alpha_1 v_1^{\otimes m},\dots, \alpha_R v_R^{\otimes m}\in \X^{\vee}_{\setft{Sep}}$ are generically chosen.

\begin{cor}\label{cor:symmetric_tensor_decomp}
Let $n$ be a positive integer, and let $m \geq 3$ be an integer. Then for a generically chosen symmetric tensor $T\in S^m(\field^n)$ of Waring rank
\ba\label{eq:R_bound_symmetric}
R \leq \min\left\{ \binom{n+\floor{m/2}-1}{\floor{m/2}}, \frac{\binom{n+\ceil{m/2}-1}{\ceil{m/2}}}{2}-\frac{\binom{n+2\ceil{m/2}-1}{2\ceil{m/2}}}{\binom{n+\ceil{m/2}-1}{\ceil{m/2}}}\right\},
\ea
in $n^{O(m)}$ time our Algorithm 1 can be used to recover the Waring rank decomposition of $T$ and certify that it is unique.
\end{cor}
Note that the bound~\eqref{eq:R_bound_symmetric} is $\Omega(n^{\floor{m/2}})$ as $n$ grows. For example, when $m=4$ the bound~\eqref{eq:R_bound_symmetric} becomes $R \leq \frac{1}{6}n(n-1)$. This matches the best known bounds for symmetric decompositions~\cite{Har72, MSS,BCPV}
(see also \cite{Vij20} for related references). Note that Corollary~\ref{cor:tensor_decomp} obtains similar bounds for non-symmetric tensors. In particular, we are not aware of any existing algorithmic guarantees (prior to our work) for generically chosen non-symmetric tensors of even $m$ that work for rank $R = \Omega(n^{ m/2 })$.

\begin{proof}
Note that the set of complex symmetric product tensors is equal to the Zariski closure of the set of real symmetric product tensors (see~e.g.~\cite[Theorem 2.2.9.2]{mangolte2020real}). Thus, by Fact~\ref{fact:real_vs_complex} it suffices to prove this statement over $\complex$. Let
\ba
\tilde{\X}_{\setft{Sep}}^{\vee} = \{v^{\otimes \ceil{m/2}} : v \in \complex^n\} \subseteq S^{\ceil{m/2}}(\complex^n)
\ea
be the set of symmetric product tensors in $S^{\ceil{m/2}}(\complex^n)$ (we omit the scalars $\alpha$ as they are redundant over $\complex$). Recall that $\X$ is non-degenerate inside of $S^{\ceil{m/2}}(\complex^n)$ and is cut out by
\ba
p=\binom{\binom{n+\ceil{m/2}-1}{\ceil{m/2}} +1}{2} - \binom{n+2\ceil{m/2}+1}{2\ceil{m/2}}
\ea
many homogeneous linearly independent polynomials of degree $d=2$. Thus, for generically chosen $ v_1^{\otimes \ceil{m/2}},\dots, v_R^{\otimes \ceil{m/2}}\in\X_{\setft{Sep}}^{\vee}$, it holds that $\{v_a^{\otimes \floor{m/2}} : a \in [R]\}$ is linearly independent, and by Theorem~\ref{cor:generic_decomp},
\ba\label{eq:symdecomp}
T=\sum_{a \in [R]} v_a^{\otimes m}
\ea
is the unique tensor rank decomposition of $T$ (and hence the unique Waring rank decomposition of $T$), and it can be recovered using Algorithm 1 in $n^{O(m)}$ time. In more details, there exists a Zariski open dense subset $\A \subseteq (\tilde{\X}_{\setft{Sep}}^{\vee})^{\times R}$ for which this holds. This translates to a Zariski open dense subset of $(\X_{\setft{Sep}}^{\vee})^{\times R}$, where
\ba
{\X}_{\setft{Sep}}^{\vee} := \{ v^{\otimes m} :  v \in \complex^n\} \subseteq S^{m}(\complex^n),
\ea
completing the proof.
\end{proof}

Finally, we can also use our framework to provide guarantees for $r$-aided rank decompositions (also known as $(r,r,1)$-block rank decompositions).

\begin{cor} \label{cor:block_decomp}
Let $n_1,n_2,n_3$ and $r < \min \{n_1,n_2\}$ be positive integers. Then 
for a generically chosen tensor $T\in {\field^{n_1}\otimes \field^{n_2} \otimes \field^{n_3}}$ of $r$-aided rank
\ba
R \leq \min\Bigg\{n_3, \frac{\binom{n_1}{r+1}\binom{n_2}{r+1}}{(r+1)! \binom{n_1 n_2+r}{r+1}}(n_1 n_2 + r)\Bigg\}=  \min\big\{n_3~, ~\Omega_r( n_1 n_2)\big\},
\ea
in $(n_1 n_2)^{O(r)}$ time our Algorithm 1 can be used to recover the $r$-aided rank decomposition of $T$ and certify that it is unique.
\end{cor}
\begin{proof}
This follows from Theorem~\ref{cor:generic_decomp}, and fact that $\X_r$ is non-degenerate of degree $r$ and is cut out by $p=\binom{n_1}{r+1}\binom{n_2}{r+1}$ linearly independent homogeneous polynomials of degree $d=r+1$ (see Section~\ref{sec:XR}).
\end{proof}



\bibliographystyle{alpha}
\bibliography{references}

\appendix

\section{Counterexample to Lemma 2.3 in~\cite{de2006link}} \label{app:counterexample}

\begin{example}[Counterexample to Lemma 2.3 in~\cite{de2006link}]
The statement of Lemma 2.3 in~\cite{de2006link} is as follows: Let $\W \subseteq \real^{n} \otimes \real^n$ be a linear subspace. Then for any positive integer $R$ satisfying $R \leq n+1$ and
\ba
\dim(\W)+\binom{R}{2} \leq n^2,
\ea
a generic collection of vectors $v_1,\dots, v_R \in \real^n$ satisfies the property that
\ba
\W \cap \spn\{v_i \otimes v_j : 1 \leq i < j \leq R\}=\{0\}.
\ea
This is false (over both $\real$ and $\complex$). Let $n=4$, let $\U \subseteq \real^4$ be an arbitrary $3$-dimensional subspace, and let $\W=\U^{\otimes 2}$. Then $R=4$ satisfies both inequalities, but for any collection of linearly independent vectors $v_1,\dots, v_4 \in \real^n$, there exist non-zero elements $u_1 \in \spn\{v_1,v_2\}\cap \U$ and $u_2 \in \spn\{v_3,v_4\} \cap \U$ (since $\U$ is a $3$-dimensional subspace of $\real^4$). It follows that
\ba
u_1 \otimes u_2 \in \W \cap \spn\{v_i \otimes v_j : 1 \leq i < j \leq 4\}.
\ea
This gives a counterexample to Lemma 2.3 in~\cite{de2006link}. The false reasoning in their proof seems to be in the fifth line of page 655 (the third to last line of the proof): Here, it seems to be implicitly claimed that for an $\real$-vector space $\V$ and three finite sets of vectors $A,B,C \in \V$, if $A \cup B$ and $B \cup C$ are linearly independent, then $\spn\{A \cup B\} \cap \spn\{B \cup C\} = \spn\{B\}$. This is incorrect (consider $A=\{e_1\}, B=\{e_1+e_2\}, C=\{e_2\}$).
\end{example}

\end{document}